\documentclass[12pt,draftcls,journal,letterpaper,twosides,onecolumn]{IEEEtran}

\usepackage[]{chapterbib} 
\usepackage{amsmath,amssymb,amscd,latexsym,dsfont,wasysym,bbm,stmaryrd}
\usepackage{float}
\usepackage{color}
\usepackage{graphicx,subfigure,balance}
\usepackage{multirow,multicol}
\usepackage{verbatim}
\usepackage{psfrag}
\usepackage{comment}
\usepackage{makeidx}
\usepackage[]{algorithm}
\usepackage{algorithmic}
\usepackage{cite}
\usepackage{arydshln} 

\usepackage{tikz}
\usetikzlibrary{shapes,arrows}

 \usepackage{subfig}
 \usepackage{amsthm}
 \usepackage{theoremref}
 \usepackage{setspace}
 
 \usepackage{multicol, blindtext}
\newcommand{\tr}[1]{\textrm{#1}}
\newcommand{\mr}[1]{\mathrm{#1}}
\newcommand{\tnr}[1]{{\textnormal{#1}}}

\newcommand{\mc}[1]{\mathcal{#1}}
\newcommand{\mf}[1]{\mathsf{#1}}

\newcommand{\ms}[1]{\mathds{#1}}

\newcommand{\ov}[1]{\overline{#1}}

\newcommand{\bk}{\boldsymbol{k}}

\newcommand{\bI}{\boldsymbol{I}}

\newcommand{\bs}{\boldsymbol{s}}

\newcommand{\bu}{\boldsymbol{u}}

\newcommand{\by}{\boldsymbol{y}}

\newcommand{\bz}{\boldsymbol{z}}

\newcommand{\figref}[1]{Fig.~\ref{#1}}

\newcommand{\secref}[1]{Sec.~\ref{#1}}

\newcommand{\tabref}[1]{Table~\ref{#1}}

\newcommand{\ie}{i.e.,~} 		
\newcommand{\eg}{e.g.,~}	

\newcommand{\argmax}{\mathop{\mr{argmax}}}

\newcommand{\set}[1]{\{#1\}}

\newcommand{\cd}{\cdot}
\newcommand{\ld}{\ldots}

\newcommand{\pdf}{p}            			

\newcommand{\cdf}{F}            			
\newcommand{\Ex}{\ms{E}}     			

\newcommand{\mcA}{\mc{A}}

\newcommand{\mcI}{\mc{I}}

\newcommand{\mcM}{\mc{M}}

\newcommand{\mcO}{\mc{O}}
\newcommand{\mcP}{\mc{P}}

\newcommand{\mcS}{\mc{S}}

\newcommand{\mcW}{\mc{W}}

\newcommand{\mcZ}{\mc{Z}}
\newcommand{\mcCr}{\mc{K}}

\newcommand{\mfa}{\mf{a}}

\newcommand{\mfm}{\mf{m}}

\newcommand{\mfs}{\mf{s}}
\newcommand{\mfz}{\mf{z}}
\newcommand{\mfw}{\mf{w}}

\newcommand{\SNR}{\mathsf{snr}}  
\newcommand{\SNRav}{\ov{\mathsf{snr}}}  

\newcommand{\SC}{\tnr{SC}}  
\newcommand{\TS}{\tnr{TS}}  

\newcommand{\R}{R}              	
\newcommand{\Nb}{{\mathop{N_\tnr{b}}}} 	
\newcommand{\Ns}{{\mathop{N_\tnr{s}}}} 	
 \newcommand{\MODMP}{\Phi^{\tnr{MP}}}  	
\newcommand{\MODwoa}{\Phi}      		

\newcommand{\ack}{\tnr{ACK}}
\newcommand{\nack}{\tnr{NACK}}
\newcommand{\kmax}{K}		
\newcommand{\kk}{k}	
\newcommand{\hol}{\ell}		

\newcommand{\Iv}{\mathbf{I}}

\definecolor{refkey}{rgb}{0.2,0.8,0.6} 
\definecolor{labelkey}{rgb}{0.2,0.3,0.4} 

\usepackage[acronym,nonumberlist]{glossaries} 
 \usepackage[usenames,dvipsnames]{pstricks}
 \usepackage{epsfig}
 \usepackage{pst-grad} 
 \usepackage{pst-plot} 
\newacronym[\glsshortpluralkey=PDFs,\glslongpluralkey=probability density functions]{pdf}{PDF}{probability density function}
\newacronym[\glsshortpluralkey=CDFs,\glslongpluralkey=cumulative density functions]{cdf}{CDF}{cumulative density function}
\newacronym[\glsshortpluralkey=PMFs,\glslongpluralkey=probability mass functions]{pmf}{PMF}{probability mass function}
\newacronym[]{cm}{CM}{coded modulation}
\newacronym[]{pam}{PAM}{pulse amplitude modulation}
\newacronym[]{bpsk}{BPSK}{binary phase shift keying}
\newacronym[]{qam}{QAM}{quadrature amplitude modulation}
\newacronym[]{psk}{PSK}{phase shift keying}
\newacronym[\glsshortpluralkey=LLRs,\glslongpluralkey=logarithmic likelihood ratios]{llr}{LLR}{logarithmic likelihood ratio}
\newacronym[]{map}{MAP}{maximum a posteriori}
\newacronym[]{ml}{ML}{maximum likelihood}
\newacronym[]{md}{MD}{multiuser diversity}
\newacronym[]{siso}{SISO}{single-input single-output}
\newacronym[\glsshortpluralkey=MIs,\glslongpluralkey=mutual informations]{mi}{MI}{mutual information}
\newacronym[\glsshortpluralkey=GMIs,\glslongpluralkey=generalized mutual informations]{gmi}{GMI}{generalized mutual information}
\newacronym[]{bicm-gmi}{BICM-GMI}{BICM generalized mutual information}
\newacronym[]{awgn}{AWGN}{additive white Gaussian noise}
\newacronym[]{amc}{AMC}{adaptive modulation and coding}
\newacronym[]{mdp}{MDP}{Markov decision process}
\newacronym[]{pomdp}{POMDP}{partially observable Markov decision process}
\newacronym[]{psimdp}{PSI-MDP}{Partial State Information Markov Decision Process}
\newacronym[]{phy}{PHY}{physical layer}
\newacronym[]{llc}{LLC}{link layer control}
\newacronym[]{sp}{SP}{set-partitioning}
\newacronym[]{gsm}{GSM}{global system for mobile communications}
\newacronym[]{edge}{EDGE}{enhanced data rates for GSM evolution}
\newacronym[]{3gpp}{3GPP}{3rd generation partnership project}
\newacronym[]{dvb}{DVB}{digital video broadcasting}

\newacronym[]{harq}{HARQ}{hybrid automatic repeat request}
\newacronym[]{irharq}{HARQ-IR}{incremental redundancy HARQ}
\newacronym[]{ir}{IR}{incremental redundancy}
\newacronym[]{rpr}{RR}{repetition redundancy}
\newacronym[]{cc}{CC}{chase combining}
\newacronym[]{hol}{HoL}{head of the line}

\newacronym[\glsshortpluralkey=PCCCs,\glslongpluralkey=parallel concatenated convolutional codes]{pccc}{PCCC}{parallel concatenated convolutional code}
\newacronym[\glsshortpluralkey=TCs,\glslongpluralkey=turbo codes]{tc}{TC}{turbo code}
\newacronym{ldpc}{LDPC}{low-density parity-check}
\newacronym[]{ofdm}{OFDM}{orthogonal frequency-division multiplexing}
\newacronym[]{bep}{BEP}{bit-error probability}

\newacronym[]{iid}{i.i.d.}{independent and identically distributed}
\newacronym[]{ami}{AMI}{accumulated mutual information}
\newacronym[]{blep}{BLEP}{block-error probability}
\newacronym[]{sep}{SEP}{symbol-error probability}
\newacronym[]{ttcm}{TTCM}{turbo-trellis coded modulation}

\newacronym[\glsshortpluralkey=SNRs,\glslongpluralkey=signal-to-noise ratios]{snr}{SNR}{signal-to-noise ratio}
\newacronym[]{ra}{RA}{resource allocation}
\newacronym[]{ral}{RA}{rate allocation}
\newacronym[]{csi}{CSI}{channel state information}
\newacronym[]{rsi}{RSI}{receiver state information}
\newacronym[\glsshortpluralkey=UFs,\glslongpluralkey=utility functions]{uf}{UF}{utility function}

\newacronym[]{er}{ER}{equal-rate}
\newacronym[]{rr}{RR}{round-robin}
\newacronym[]{pf}{PF}{proportional fairness}
\newacronym[]{ts}{TS}{time-sharing}
\newacronym[]{sc}{SC}{superposition coding}

\newacronym[]{pf-ts}{PF-TS}{proportionally fair TS}
\newacronym[]{pf-sc}{PF-SC}{proportionally fair SC}
\newacronym[]{er-ts}{ER-TS}{equal-rate TS}
\newacronym[]{er-sc}{ER-SC}{equal-rate SC}
\newacronym[]{hm}{HM}{hierarchical modulation}

\newacronym[]{lhs}{l.h.s.}{left-hand side}
\newacronym[]{rhs}{r.h.s.}{right-hand side} 
\newacronym[\glsshortpluralkey=BSs,\glslongpluralkey=base-stations]{bs}{BS}{base-station}
\newacronym[\glsshortpluralkey=MSs,\glslongpluralkey=mobile-stations]{ms}{MS}{mobile-stations}
\newacronym[]{kkt}{KKT}{Karush--Kuhn--Tucker} 
\newacronym[]{mcs}{MCS}{modulation/coding scheme}

\newcommand{\siz}{1.2}  
\newcommand{\sizf}{0.65}   

\IEEEoverridecommandlockouts

\title{Multi-packet Hybrid ARQ: Closing gap to the ergodic capacity}

\author{Mohammed Jabi,~Aata El Hamss,~Leszek Szczecinski,~and~Pablo Piantanida
\\
\thanks{%
M. Jabi and L. Szczecinski are with INRS-EMT, Montreal, Canada. [e-mail: \{jabi,~leszek\}@emt.inrs.ca].}%
\thanks{%
A. El Hamss is with InterDigital Communications, Montreal, Canada. [e-mail: aata.elhamss@gmail.com].}%
\thanks{%
P. Piantanida is with SUPELEC, Gif-sur-Yvette, France. [e-mail: pablo.piantanida@supelec.fr].}%
 \thanks{%
Part of this work will be presented at the {IEEE} Global Communication Conference (GLOBECOM), Austin, Texas, USA, December 2014.} %
\thanks{%
The work was supported by the government of Quebec, under grant \#PSR-SIIRI-435.}%
}%

\newtheorem{lemma}{Lemma}

\begin{document}

\maketitle
\thispagestyle{empty}

\begin{abstract}
In this work we consider \gls{ir} \gls{harq}, where transmission rounds are carried out over independent block-fading channels. We propose the so-called multi-packet HARQ where the transmitter allows different packets to share the same channel block. In this way the resources (block) are optimally assigned throughout the transmission rounds. This stands in contrast with the conventional HARQ, where each transmission round occupies the entire block. We analyze superposition coding and time-sharing transmission strategies and we optimize the parameters to maximize the throughput. Besides the conventional one-bit feedback (ACK/NACK) we also consider the rich, multi-bit feedback. To solve the optimization problem we formulate it as a Markov decision process (MDP) problem where the decisions are taken using accumulated mutual information (AMI) obtained from the receiver via delayed feedback.  When only one-bit feedback is used to inform the transmitter about the decoding success/failure (ACK/NACK), the Partial State Information Markov Decision Process (PSI-MDP)  framework is used to obtain the optimal policies.  Numerical examples obtained in a Rayleigh-fading channel indicate that, the proposed multi-packet HARQ outperforms the conventional one, by more than 5 dB for high spectral efficiencies.
\end{abstract}


\section{Introduction}\label{Sec:Introduction}
\IEEEPARstart{I}{n this} work we propose and  optimize the strategies to increase the throughput of \gls{harq} transmission over a block fading channel. 

\gls{harq} is mostly used to deal with the loss of data packets due to unpredictable changes in the channel and consists in ``handshaking''  between the transmitter and the receiver: the receiver sends the binary messages via a feedback channel to inform the transmitter about the success (ACK message) or failure (NACK) of the transmission. Then  a new version of the lost packet is transmitted. This process continues till ACK is received or---in the case of \emph{truncated} \gls{harq}---till the maximum number of transmission rounds is reached.  In practice, the truncation appears as an implementation constraint but also may be justified when transmitting data that is delay-sensitive, which after a prescribed time may loose its validity.

\gls{harq} is often perceived as an additional ``guarantee'', which works on the top of the \gls{phy}. However, it was already shown in previous works, that adjusting the \gls{phy}-related parameters (rate, coding, power) as a function of the \gls{harq} state can significantly improve the performance, \eg in terms of the average transmission rate (throughput) \cite{Cheng03,Uhlemann03,Visotsky05,Pfletschinger10,Szczecinski13}, or the outage \cite{LiuR03,Kim11,Nguyen12}.

While the parameters of the PHY may be adjusted as a function of the number of received NACK messages, more substantial gains are obtained exploiting the additional information sent by the receiver. We thus consider the case when, on top of ACK/NACK messages, the receiver conveys over the feedback channel, ``multi-bit'' \gls{rsi}. As we explain later, \gls{rsi} represent the state of the decoder, which is related to the \gls{csi} determined by the channel gains in different blocks. These channel gains are assumed to be \gls{iid} random variables (a modelling, which captures adequately the fact that retransmissions are usually scheduled in well-spaced time instant to absorb the processing and round-trip delays). Therefore, the transmitter is unable to infer the instantaneous \gls{csi} from the \gls{rsi}, and thus, the conventional \gls{amc} cannot be used.

Nevertheless, the \gls{rsi} that accompanies the NACK message, provides the transmitter with valuable prior allowing it to suitably adjust its \gls{phy} parameters in the subsequent rounds. Indeed, the use of \gls{rsi} was already considered before in \cite{Visotsky05,Gan10,Szczecinski13} to adapt the length of the codewords with the implicit objective of shortening the average number of transmissions and using the time ``released'' in such a way to send more packets. 

It is worth mentioning here that, in \cite{Visotsky05,Gan10,Szczecinski13}  the lengths were optimized in abstraction of the system-level considerations. Thus, exploiting the released time is not obvious because the codewords with variable lengths are sent within the same block. The lengths are optimized in abstraction of the block-length so, in general, they do not match the latter. This, will produce an unoccupied time within the block, which in turn translates into a throughput loss. To address this problem, \cite{Wang07,Szczecinski13} proposed the use of many packets within a single block,  reducing in this way the impact of the unoccupied time. This may go, however,  against the practical considerations of having only one or a few packets in the block. Moreover, regardless of system-level considerations, recognizing that various packets are transmitted within the same block raises the question of optimality of the conventional \gls{ts} approach, and the \gls{sc} \cite{Shamai03} may be an alternative as already considered before with \gls{harq}, \eg in \cite{Zhang09,Takahashi10,Chaitanya11b}.

\begin{figure*}[th]
\begin{center}
\pgfdeclarelayer{background}
\pgfdeclarelayer{foreground}
\pgfsetlayers{background,main,foreground}


\tikzstyle{form1} = [draw, minimum width=1.6cm, text width=1.5cm, fill=blue!5, 
  text centered,  minimum height=.9cm]

\tikzstyle{form2} = [draw=none, minimum width=3.4cm, text width=1.5cm, fill=none, 
  text centered,  minimum height=0cm]    

\tikzstyle{form3} = [draw,dashed, minimum width=1.5cm, text width=1.4cm, fill=blue!2, 
  text centered,  minimum height=1.8cm]
  
  \tikzstyle{form4} = [draw=none, minimum width=0cm, text width=0cm, fill=none, 
  text centered,  minimum height=.6cm] 
    
    \tikzstyle{elip}=[draw,ellipse, minimum width=.2cm, text width=0cm, fill=none, 
  text centered,  minimum height=1cm]
  
    \tikzstyle{elip2}=[draw,ellipse, minimum width=.1cm, text width=0cm, fill=none, 
  text centered,  minimum height=.6cm]

\begin{tikzpicture}

    \node(Channel)[form1] {Channel};
    
     \path (Channel.west)+(-2.7,0)  node[form1](Enc) {Encoder};
     \path (Channel.east)+(2.7,0)  node[form1](Dec) {Decoder};
     \path (Channel.north)+(0,1.6)  node[form3](FedBack) {Feedback Channel};
      \path (Enc.north)+(0,1.6)  node[form1](Arq_cont1) {HARQ Controller};
      \path (Dec.north)+(0,1.6)  node[form1](Arq_cont2) {HARQ Controller};
       \path (Enc.west)+(-1,-0.2)  node(Buff1) {};
       
       \path (Buff1)+(0,2.3)  node[form2](Buff11) {};
          \path (Buff1)+(-2,2.3)  node[form2](Buff13) {};
        
         \path (Buff1.south)+(0,-.3)  node(Buff1P1) {};
         \path (Buff1.north)+(0,.6)  node[form4](Buff1P2) {};
                    
        \path (Dec.east)+(1,0)  node[form2](Buff2) {};
        
         \path (Buff2)+(0,2.3)  node[form2](Buff21) {};
          \path (Buff2)+(2,2.3)  node[form2](Buff23) {};
        
           \path (Buff2.south)+(0,-.2)  node(Buff2P1) {};
         \path (Buff2.north)+(0,.2)  node(Buff2P2) {};
         \path (FedBack.east)+(0,-0.6)  node(P1_fed) {};
          \path (FedBack.west)+(0,-0.7)  node(P2_fed) {};
           \path (Dec.north)+(-0.6,0)  node(P1_Dec) {};
           \path (Enc.north)+(0.6,0)  node(P1_Enc) {};


      \path[draw] (Buff11)--(Buff1P1.base);
         \path[draw] (Buff1P1)+(-2,0)--(Buff13);

    \path[draw] (Buff21)--(Buff2P1.base);
         \path[draw] (Buff2P1)+(2,0)-- (Buff23);

       \path[draw] (Buff1.north)+(-2,0)--(Buff1.north);
     \path[draw] (Buff1P1)+(-2,0)--node[shift={(0,.3)}]{$\bu_{\hol}$}(Buff1P1.base);
      \path[draw] (Buff1P2)+(-2,0)--node[shift={(0,-.3)}]{$\bu_{\hol+1}$}(Buff1P2.base);
      
       \path[draw] (Buff2P1)+(2,0)--(Buff2P1.base);
      \path[draw] (Buff2P2)+(2,0)--(Buff2P2.base);

            \path (Buff1P1)+(-1,1.9) node(pts1){$\vdots$};

      \path (Buff2P1)+(1,1.3) node(pts4){$\vdots$};
      
      \path(Buff1P1)+ (.5,0.6) node[elip](el){};
      
      \path(Buff2P1)+ (-.5,0.3) node[elip2](el2){};
      

 
 \path[draw,->,>=latex] (Enc.east)--node[shift={(0,.3)}]{$\bs[n]$}(Channel.west);
  \path[draw,->,>=latex] (Channel.east)--node[shift={(0,.3)}]{$\by[n]$}(Dec.west);

  \path[draw,dashed,<-,>=latex] (Buff1P2.south)+(1,0)--(Buff1P2.south);
   \path[draw,<-,>=latex] (Buff1.base)+(1,0)--(Buff1.base);
     \path[draw,->,>=latex] (Buff2.base)+(-1,0)--(Buff2.base);

    \path[draw,->,>=latex] (Arq_cont1.south)--node[shift={(0.8,0)}]{$\mfa[n]$}(Enc.north);
     \path[draw,<-,>=latex] (Arq_cont2.south)--node[shift={(0.4,0)}]{$\mcM[n]$}(Dec.north);

      \path[draw] (Arq_cont2.west)--node[shift={(0,.3)}]{$\mcM[n]$}(FedBack.east);
      
       \path[draw,<-,>=latex] (Arq_cont1.east)--node[shift={(0,.3)}]{$(\mcM[n], I[n])$}(FedBack.west);
       
       \path[draw] (P1_fed.base)-| node[shift={(-1.2,.2)}]{$I[n]$}(P1_Dec.base);
       
       \path[draw,->,>=latex] (Arq_cont1.west)-|(el.north);
       \path[draw,->,>=latex] (Arq_cont2.east)-|(el2.north);


\end{tikzpicture}

\end{center}
\caption{Model of the \gls{harq} transmission: the \gls{harq} controller uses the information obtained over the feedback channel to take actions $\mfa[n]$, which determine the encoding.}\label{Fig:HARQ_buffer}
\end{figure*}
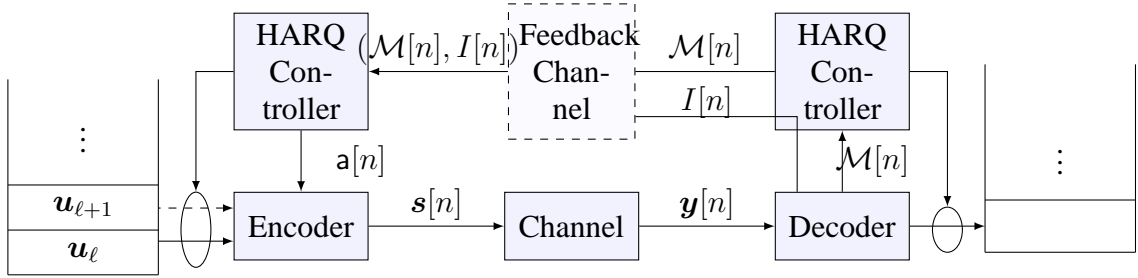

In this work we address these two issues. On one hand, we explicitly consider the constraints resulting from the transmission of variable-length packets within the same block \cite{Verdu10}; this links the \gls{harq} design with system-level considerations, an issue that was lacking in \cite{Visotsky05, Szczecinski13}. On the other hand, we consider the \gls{sc} and \gls{ts} as alternatives for the joint encoding of the new and retransmitted packets; we formally optimize the parameters which was not considered in \cite{Zhang09,Takahashi10,Chaitanya11b}.

We look at the \gls{harq} as a \emph{control} process based on the feedback signal and  we optimize the \emph{actions} which are given by the joint encoding (and its parameters) to be used. To solve the optimization we define our problem as a \gls{mdp} \cite[Ch.~7.4]{Bertsekas05_book}, which was already used in different contexts for \gls{harq} optimization, \eg in  \cite{Karmokar06,Djonin08}. We maximize the  throughput which is a relevant performance criterion as it can be directly related to the ergodic (long term) channel capacity \cite{Caire01,Wu10}. 

The paper is organized as follows: \secref{Sec:Model}  defines the model of the system under study, which is then cast into \gls{mdp} in \secref{Sec:MDP.Signalling}. We explain the optimization in \secref{Sec:Th.MDP} and the numerical results are shown and discussed in \secref{Sec:Num.Results}. The case of one bit feedback is analyzed in \secref{Sec:One_bit}. We conclude the work in \secref{Sec:Conclusions}.

\section{Model of \gls{harq}}\label{Sec:Model}
We first describe the conventional \gls{harq}, which allows us to define the useful notation and next we discuss the \emph{multi-packet} \gls{harq} we propose.
\subsection{Conventional HARQ and Notations}
We consider a point-to-point transmission using \gls{harq} shown in \figref{Fig:HARQ_buffer}, where  the transmitter sends the data block $\bs[n]$ of $\Ns$ symbols over a block-fading channel. Each block $n$ contains the encoded version of $\Nb$ information bits contained in $\bu_{\hol}$, where $\hol$ indicates the \gls{hol} packet at block time $n$ (the \gls{hol} packet is the first packet in the buffer to be transmitted); for convenience we say that the packet $\bu_{\hol+1}$ is ``\gls{hol}-next''. The coding rate per packet is thus $\R=\Nb/\Ns$. The receiver observes the channel outcome $\by[n]$ and attempts to decode $\bu_{\hol}$. We note that the indices of the packets are not the same as indices of the blocks, we refer to the former via subindexing, \eg $\bu_\hol$, and--to the latter--via arguments within brackets, \eg $\bs[n]$. We also use $n_{\ell}$ to denote the index of the block $s[n]$ when the packet $\bu_{\ell}$ was transmitted for the first time. 


More than one transmission may be necessary to deliver the packet and the index of the transmission round of  the \gls{hol} packet is kept by transmitter in the \gls{harq} counter $\kk_{\hol}$. For convenience, we assume that the \gls{harq} counter $\kk$ for the $t$-th packet entering the transmitter's buffer, where $t>\hol$ is set to zero, \ie \mbox{$\kk_{t}\leftarrow 0$}. The decoding errors are detected at the receiver which sends to the transmitter a binary message $\mcM[n]=\mcM_{\ell,\kk_\ell}$, where $\mcM_{\ell,\kk_\ell}=\ack$ if the decoding of $\bu_{\hol}$ is successful in the $\kk_\ell$-th round,  or $\mcM_{\ell,\kk_\ell}=\nack$ if the decoding fails. Reassume that the feedback channel is error-free.

Due to the propagation, transmission, and processing delays, the message $\mcM[n]$ arrives at the transmitter at time block $n+1$, where it can be used by the \gls{harq} controller. The latter discards the \gls{hol} packet when $\ack$ is received or the maximum number of transmission rounds $\kmax$ is reached; then the \gls{hol}-next packet becomes the \gls{hol}. Formally, this is done, first, incrementing the \gls{hol} index, which otherwise does not change, \ie 
\begin{equation}\label{upd.rule}
\hol\leftarrow
\begin{cases} 
\hol+1 , &\text{if}	~\mcM_{\ell,\kk_\ell}=\ack~\text{or}~\kk_{\hol} = \kmax \\
\hol, 	& \text{otherwise }
\end{cases}
\end{equation} 
and, next, increasing the \gls{harq} counter of the \gls{hol} packet
\begin{align}\label{upd.rule.1}
\kk_{\hol}\leftarrow \kk_{\hol}+1.
\end{align}

When the packets have their counter sets  to zero, increasing it via \eqref{upd.rule.1}, we obtain $\kk_\hol=1$, which means that we start the \gls{harq} process of the packet $\bu_\hol$.

The encoding, in general depend on the index of the transmission round, \ie on the \gls{harq} counter $\kk_\hol$, \ie 
\begin{align}\label{Enc.1}
\bs[n]=\MODwoa_{\kk_\hol}(\bu_{\hol}).
\end{align}

In particular, when $\MODwoa_{\kk}(\bu_{\hol})=\MODwoa_{1}(\bu_{\hol}), \kk=1,\ld,\kmax$ we have the case of transmission with the \gls{rpr} \cite{Benelli85} \cite{Chase85}. That is, irrespectively of $\kk_\hol$, the transmitted symbols $\bs[n]$ are always the same for the given $\bu_\ell$.  If, on the other hand, $\kk_\hol$ is used to extract different subcodewords of the mother code's codeword
\begin{align}\label{}
[\MODwoa_{1}(\bu),\ld, \MODwoa_{\kmax}(\bu)],
\end{align}
we obtain the well-known \gls{irharq} \cite{Hagenauer88}, where conventionally, all the subcodewords $\MODwoa_{k_{\hol}}$ have the same length. \gls{irharq} is the focus of this work.

The  channel outcome at the receiver is given by
\begin{equation}
\by[n]=\sqrt{\SNR[n]}\bs[n]+\bz[n],
\end{equation}
where  $\SNR[n]$ is the channel \gls{snr}, $\bz[n]$ is a zero mean, unitary-variance Gaussian variable modelling noise. We assume that $\SNR[n]$ are \gls{iid} and their \gls{pdf} is known. For numerical evaluation we use the exponential form (\ie we consider Rayleigh block fading)
\begin{equation}
\pdf_{\SNR}(x)=\SNRav^{~-1}\exp\big(- {x}/{\SNRav}\big),
\end{equation}
where $\SNRav$ is the average \gls{snr}. 

In \gls{irharq}, the receiver decodes the packet $\bu_{\hol}$ concatenating the $\kk\leq\kmax$ blocks of channel outcomes
\begin{align}\label{}
\by_{\hol,\kk}\triangleq\big[\by[n_\hol],\ld,\by[n_\hol+\kk-1]\big].
\end{align}

Then, assuming the subcodewords $\MODwoa_{\kk}(\bu_{\hol})$ are drawn from randomly generated codebook, and for sufficiently large $\Ns$, the decoding is successful if the \gls{ami} defined as 
\begin{align}\label{ami.defi}
I_{\hol,\kk}\triangleq\sum_{t=0}^{k-1} C(\SNR[n_\hol+t]) 
\end{align}
exceeds the transmission rate, \ie
\begin{align}\label{ami.k}
I_{\hol,\kk}>\R.
\end{align}
For simplicity, we assume only  Gaussian inputs, \ie
\begin{align}\label{}
C(\SNR)\triangleq\log_2(1+\SNR).
\end{align}
The \gls{ami} is thus equivalent to the \gls{rsi} and we may sent it to the transmitter over the feedback channel but it is irrelevant for the conventional HARQ.
\subsection{Multi-packet HARQ}\label{subsec:S_Model}
In the conventional \gls{harq},  $\bs[n]$ depends solely on the \gls{harq} counter $\kk_{\hol}$ and the packet contents $\bu_{\hol}$, as per \eqref{Enc.1}. For the proposed multi-packet HARQ, we assume that the encoder is able to jointly encode the \gls{hol} packet $\bu_{\hol}$ and the \gls{hol}-next packet $\bu_{\hol+1}$ as shown also schematically in \figref{Fig:HARQ_buffer}, \ie
\begin{align}\label{}
\bs[n]=\MODMP({\bu_{\hol},\bu_{\hol+1},\kk_{\hol},\kk_{\hol+1}, p}),
\end{align}
where $p$ is the parameter of the encoding and it is a function of the counters $\kk_{\hol}$ and $\kk_{\hol+1}$ and the \gls{rsi}.

We consider four possible encoding modes
\begin{itemize}
\item 
Conventional, one-packet transmission, which we denote by '1P', where the block $\bs[n]$ is occupied only by the subcodewords of the \gls{hol} packet,  
\item
Dropped, one-packet transmission, which we denote by '0P', where we stop the transmission of the \gls{hol} packet and $\bs[n]$ is occupied only by the subcodewords of the \gls{hol}-next packet,
\item
\gls{ts} transmission, where the subcodewords of the \gls{hol} and \gls{hol}-next packets are transmitted in non-overlaping parts of the block, with the time-sharing defined by $p$, and  
\item 
\gls{sc} transmission, where both codewords are superimposed with power fractions defined by $p$. 
\end{itemize}
We clarify this in \tabref{Tab:mod.schemes}, where we add details to the encoding notation in \eqref{Enc.1} as follows:
\begin{align}\label{}
\MODwoa_{\kk_\hol}(\bu_\hol,\alpha)
\end{align}
where $\alpha$ indicates the relative length of the $\kk$th subcodeword composed of $\alpha \Ns$ symbols.

\begin{table}[ht!]
\centering
\begin{tabular}{l|c}
$\mfa[n]$  	& $\MODMP({\bu_{\hol},\bu_{\hol+1},\kk_{\hol},\kk_{\hol+1}, p})$ \\
\hline
\hline
$(\tnr{1P}, -)$  	& $\MODwoa_{\kk_{\hol}}(\bu_\hol, 1)$\\
\hline
$(\tnr{0P}, -)$  	& $\MODwoa_{\kk_{\hol+1}}(\bu_{\hol+1}, 1)$\\
\hline
$(\TS, p)$  	& $\big[ \MODwoa_{\kk_{\hol}}(\bu_\hol, p), \MODwoa_{\kk_{\hol}+1}(\bu_{\hol+1}, 1-p) \big]$\\
\hline
$(\SC, p)$ 	& $\sqrt{p}\MODwoa_{\kk_{\hol}}(\bu_\hol, 1)+\sqrt{1-p} \MODwoa_{\kk_{\hol+1}}(\bu_{\hol+1}, 1)$
\end{tabular}
\caption{Results of encoding actions $\mfa[n]$ for the joint encoding of the \gls{hol} and \gls{hol}-next packets $\bu_\hol$ and $\bu_{\hol+1}$.}
\label{Tab:mod.schemes}
\end{table}

We note that when the joint encoding (\gls{ts} or \gls{sc}) is chosen, the feedback channel transmits the decoding result for both \gls{hol} and \gls{hol}-next packets, \ie $ \mcM[n]=(\mcM_{{\hol},\kk_{\hol}},\mcM_{{\hol+1},\kk_{\hol+1}})$.

The role of the \gls{harq} controller is to decide on the encoding \emph{actions} $\mfa[n]=(\mfm[n],p[n])$, where $\mfm[n]\in\mcA_\tnr{mod}=\set{\tnr{1P},\tnr{0P},\TS,\SC}$ defines the  encoding ``mode'' and $p[n] \in \mcA_{\tr{p}}=]0,1[$ is the encoding parameter as specified also in \tabref{Tab:mod.schemes}. The conventional \gls{harq} always takes the same encoding action $\mfa[n]=(\tnr{1P},-)$, where we use ``$-$''  to indicate that, in this case, the parameter $p$ is irrelevant from the point of view of the encoding. The actions are taken from the action-space $\mcA=\mcA_\tnr{mod}\times \mcA_\tr{p}$ and may result in one-packet (\tnr{1P} or \tnr{0P}) or a  multi-packet (\gls{sc} or \gls{ts}) transmissions.

The rule for updating $\hol$ becomes now the following:
\begin{equation}\label{upd.rule.joint}
\hol\leftarrow
\begin{cases} 
\hol+1 , 	&\text{if}~\mcM_{\hol,\kk_{\hol}}=\ack~\text{or}~\kk_{\hol} = \kmax~\text{or}~\mfm=\tnr{0P}\\
\hol, 	& \text{otherwise }
\end{cases}
\end{equation} 
and the \gls{harq} counters' update takes into account the possibility of joint encoding
\begin{align}\label{upd.rule.joint.1}
\kk_{\hol}\leftarrow \kk_{\hol}+1,~~~~~~~\forall \mfm\in \mcA_\tnr{mod},~~~~~~
\end{align} 
\begin{align}\label{upd.rule.joint.2}
\kk_{\hol+1}\leftarrow
\begin{cases} 
0,&\text{if}~\mfm\in\set{\tnr{1P}\cup\tnr{0P}}\\
\kk_{\hol+1}+1,&\text{if}~\mfm\in\set{\TS\cup\SC}
\end{cases} ,
\end{align}
where the second condition in \eqref{upd.rule.joint.2} reflects the fact that both, \gls{hol} and \gls{hol}-next packets are transmitted simultaneously.

We note that the actions from the subspace $\mcA_\tr{mod} \times \set{0,1}$ are explicitly excluded because for the joint encoding with $p\in\set{0,1}$, one of the packets is deprived of the transmission time (TS) or power (SC), while its \gls{harq} counter would increment as per the second line of \eqref{upd.rule.joint.2}. This is clearly suboptimal so the cases of $p\in\set{0,1}$ are handled by the encoding modes \tnr{1P} or \tnr{0P}, where only one \gls{harq} counter is incremented. 

We also note that the action $(\tnr{0P},-)$ was already used in \cite{Szczecinski13}\cite{jabi14} and means that the packet is abandoned when there is no ``reasonable hope'' to decode it successfully.  While making such a decision may seem drastic, the dropped packeted may be re-injected into the transmitter's buffer as may also be other packets considered lost. This issue depends on the sensitivity of the source to the delay in the packets' delivery---the problem we do not consider here.

The decoding is done in the similar way as in the one-packet \gls{harq}. In the case of the \gls{ts} encoding, after the $k$-th transmission, the decoding is successful provided that \mbox{$I_{\hol,k}>R$}, with 
\begin{align}\label{I.ts.1}
I_{\hol,k}=I_{\hol,k-1}+  p C(\SNR),
\end{align}
where  by definition $I_{\hol,0}=0$ and to simplify the notation we used $\SNR\equiv\SNR[n_\hol+k-1]$.

At the same time we find the \gls{ami} for the \gls{hol}-next packet as
\begin{align}\label{I.ts.2}
I_{\hol+1,k}=I_{\hol+1,k-1}+(1-p)C(\SNR).
\end{align}

Since $\SNR$ is random, increasing $p$, the probability of successful decoding of the packet $\bu_{\hol}$ increases  with $p$, see \eqref{I.ts.1}, but at the same time, the probability of correct decoding of the \gls{hol}-next packet  is decreased. Thus, it is not possible to improve simultaneously the reliability of transmission of both packets. The challenge of optimizing  the encoding actions lies in striking the balance between these  contradictory effects.

In the case of the \gls{sc}, the decoding is slightly more involved because the \gls{hol} and \gls{hol}-next packets interfere with each other. For simplicity we only consider single packet decoding, in contrast to  joint decoding of both packets\footnote{While joint decoding is possible (as we transformed our point-to-point channel into multiple access channel), joint decoding would result in an additional complex at the receiver and on the other hand the analysis would become much involved.}. That is, we assume that decoding of $\bu_\hol$ depends solely on $I_{\hol,k-1}$ and $\by[n_\hol+k-1]$, \ie the \gls{ami} of the \gls{hol} packet is given by 
\begin{align}\label{I.sc.1}
I_{\hol,k}=I_{\hol,k-1}+  C\left(\frac{p\SNR}{1+(1-p)\SNR}\right).
\end{align}

The \gls{ami} for the \gls{hol}-next packet is given by 
\begin{align}\label{I.sc.2}
I_{\hol+1,k}=
\begin{cases}
I_{\hol+1,k-1}+C\left((1-p) \SNR \right)&\text{if}\quad I_{\hol,k}> R\\
I_{\hol+1,k-1}+C\left(\frac{(1-p) \SNR}{1+p\SNR}\right)&\text{if}\quad I_{\hol,k}\leq R
\end{cases};
\end{align}
in the case $\{I_{\hol,k}> R\}$  we assume that the interference induced by the superposed subcodeword $\MODwoa_{\kk_{\hol}}(\bu_\hol)$ was removed; in the other case---that the interference cannot be removed because the \gls{ami} related to the \gls{hol} packet is not sufficiently large to allow for decoding ($I_{\hol,k}\leq R$).

\section {Multi-bit feedback} \label{sec_multi.bit}

We assume that after each  multi-packet \gls{harq} round, on top of the conventional signalling $\mcM[n]=(\mcM_{{\hol},\kk_{\hol}},\mcM_{{\hol+1},\kk_{\hol+1}})$, the transmitter is provided with additional information about the state of the receiver. In particular, since the decoding success/failure are determined by the AMI, see \eqref{ami.k}, the \gls{ami} $\bI[n]=(I_{{\hol},\kk_{\hol}},I_{{\hol+1},\kk_{\hol+1}})$  is sent via the feedback channel to the transmitter\footnote{We may send $I_{\ell,\kk_\ell}$ or, instead, report $\SNR[n]$ and let the transmitter to calculate the \gls{ami} via \eqref{I.ts.1}, \eqref{I.ts.2}, \eqref{I.sc.1}, or  \eqref{I.sc.2}}, as shown in \figref{Fig:HARQ_buffer}. We focus on optimizing the encoding action $\mfa$, targeting the maximization of the throughput. The key idea is to represent the multi-packet \gls{harq} as an \gls{mdp} $(\mcS, \mcA,  \mcW, Q, r )$ where  $\mcS,  \mcA$ and $\mcW$ are the state space, the action space and the disturbance space respectively, $Q$ is the transition law and $r$ is the reward  \cite[Ch.~1]{Bertsekas05_book}. 

Being in state $\mfs \in \mcS$ at block time $n$ gives the controller the possibility to take action $\mfa[n] \in \mcA(\mfs)$,  after which the \gls{harq} moves to the next state $\mfs'\in\mcS$ at time $n+1$.  In general, not all actions are allowed in a state $\mfs$, thus  $ \underset{\mfs \in \mcS}{\bigcup}\mcA(\mfs)=\mcA$. The transition from the state $\mfs$ to $\mfs'\in\mcS$ depends uniquely on action $\mfa[n]$ and on a random disturbance $\mfw[n]\in\mcW$. Thus, the transition between states is characterized in probabilistic manner described by the transition law  $Q$. The function $r(\mfs,\mfa): \mcS\times \mcA\rightarrow \mathbb{R}$ defines the reward acquired when the system is in state $\mfs$ and the controller chooses action $\mfa$.  
\subsection{Markov Decision Process}\label{Sec:MDP.Signalling}
In the multi-packet \gls{harq}, the feedback messages \gls{ami} $\bI[n]$, as well as the \gls{harq} counters $\bk[n]=(\kk_{\hol}, \kk_{\hol+1})$ define the state of the \gls{harq} process. The \gls{ami} $\bI$ is defined over the set $\mcI^{2}$ by the quantities:
\begin{align}
\label{def.mcI}
\mcI&\triangleq[0,R]\cup \mcI_{R_+},\\
\label{def.Rp}
\mcI_{R_+}&\triangleq(R,\infty),
\end{align}
where \eqref{def.Rp} explicitly groups  those values of the \gls{ami} which lead to the event of decoding success ($I_{\hol,k}>R$).
With $\kmax$ allowed transmission rounds and considering that \gls{hol} and \gls{hol}-next packets may be sent simultaneously, the \gls{harq} counters $\bk$ can only take values in the set $\mcCr=\set{(1,0),(2,1),(2,0),\ldots,(K,K-1),(K,K-2),\ldots,(K,0)}$. Thus, the state space $\mcS$ is defined as: $\mcS=\mcCr \times \mcI^{2}$. Each state is thus represented by a quadruplet $(\kk_{\ell},\kk_{\ell+1},I_{\ell},I_{\ell+1})$. We note that the value of \gls{ami} $\bI[n]$ contains implicitly the results of the decoding $\mcM[n]$.

For $\mfs \in \mcS$, we denote by  $\mcM^{\mfs}$, $\bk^{\mfs}$ and $\bI^{\mfs}$ the corresponding value of $\mcM$,  $\bk$ and $\bI$ respectively. We distinguish the following subsets:
\begin{align}
\mcS_{\ack,\ack}&=\Big\{\mfs\in \mcS | \mcM^{\mfs}=(\ack,\ack)\Big\}\\
\mcS_{\nack,\nack}&=\Big\{\mfs\in \mcS | \mcM^{\mfs}=(\nack,\nack)\Big\}\\
\mcS_{\ack,\nack}&=\Big\{\mfs\in \mcS | \mcM^{\mfs}\in\set{(\ack,\nack), \nonumber \\
  &~~~~~~~~~~~~~~~~~~~(\nack,\ack) }\Big\}\\
 \mcS^{\tnr{1P}}_{\ack}&=\Big\{\mfs\in \mcS | \set{\mcM^{\mfs}=(\ack,-) \nonumber \\
  &~~~~~~~~~~~~~~~~~\cap \bk^{\mfs}=(l,0), 1\leq l\leq\kmax }\Big\}\\
  \mcS^{\tnr{1P}}_{\nack}&=\Big\{\mfs\in \mcS | \set{\mcM^{\mfs}=(\nack,-) \nonumber \\
  &~~~~~~~~~~~~~~~~~\cap \bk^{\mfs}=(l,0), 1\leq l\leq\kmax }\Big\}
 \end{align}
 where $\mcS_{\ack,\ack}$ contains states when both packets are decoded correctly while $\mcS_{\nack,\nack}$ presents states when both packets failed to be decoded, $\mcS_{\ack,\nack}$ is the set of states when only one packet is decoded correctly,  $\mcS^{\tnr{1P}}_{\ack}$ and $\mcS^{\tnr{1P}}_{\nack}$ characterize the states when the one-packet transmission mode is adopted in the whole \gls{harq} rounds and that the packet is successfully or unsuccessfully decoded respectively.

We also consider three types of \gls{harq} differentiated by their respective action space:
\begin{enumerate}
\item Conventional one-packet \gls{harq} (\tnr{1P}) when $p$ is irrelevant. In this case the action space $\mcA$ has only one element $(\tnr{1P}, -)$ and no optimization is needed. 
\item Time sharing multi-packet mode (\TS) when the action space is defined as $\mcA=\mcA_\tnr{mod}^{\TS}\times \mcA_{\tr{p}}$ with $\mcA_\tnr{mod}^{\TS}=\set{\tnr{1P},\tnr{0P},\TS}$.
\item Superposition coding multi-packet mode (\SC) when $\mcA=\mcA_\tnr{mod}^{\SC}\times \mcA_{\tr{p}}$ with $\mcA_\tnr{mod}^{\SC}=\set{\tnr{1P},\tnr{0P},\SC}$.
\end{enumerate}


We assume that the joint-encoding actions, \ie $\mfm\in \set{{\TS}, {\SC}}$, can be taken for all states $\mfs \in \mcS$ except if  $\mfs \in \set{\mcS_{\ack,\ack} \cup \mcS^{\tnr{1P}}_{\ack} }$ when only one packet is transmitted. That is, the joint-encoding actions is only allowed when the \gls{hol} packet is not succesfully decoded.

The statistically evolution of the system is represented by the transition law $Q$. Since the states and the actions are discrete, see \secref{sec.discri}, this law is given by the state-to-state transition probabilities:
\begin{equation}\label{p.ssa}
p_{\mfs,\mfs'}(\mfa)\triangleq\Pr\set{ \mfs[n+1]=\mfs'|\mfs[n]=\mfs, \mfa[n]=\mfa}.
\end{equation}
which can be calculated using \eqref{I.ts.1}-\eqref{I.sc.2}. An example of $p_{\mfs,\mfs'}$ calculation is shown in Appendix A. We assume here the stationary behaviour, in which the time $n$ is irrelevant for the transition probability; this condition is satisfied for sufficiently large $n$, \ie after the transients phase. 

The transition from the state $\mfs[n]$ to the state $\mfs[n+1]$ depends on the action $\mfa[n]$ and on the disturbance $\SNR[n] \geq 0$ as can be seen in \eqref{I.ts.1}-\eqref{I.sc.2}. Consequently we take $\mcW=\mathbb{R}^{+}$. 

Each state-transition yields the reward given by
\begin{align}
\hat{r}(\mfs,\mfa,\mfs')=
\begin{cases} 
R , &\mfs' \in \set{\mcS_{\ack,\nack},\mcS_{\nack,\ack}, \mcS^{\tnr{1P}}_{\ack}} \\
2R, &\mfs' \in \mcS_{\ack,\ack}\\
0, & \text{otherwise},
\end{cases}
\end{align}
and the expected reward for taking action $\mfa$ in the state $\mfs$ is then given by 
\begin{equation}\label{exp.rew}
r(\mfs,\mfa)=\sum_{\mfs'\in\mcS} p_{\mfs,\mfs'}(\mfa)\hat{r}(\mfs,\mfa,\mfs').
\end{equation}

Our objective is thus to 
find a \emph{policy} $\pi: \mcS\mapsto\mcA$, such that the \gls{harq} controller taking actions $\mfa=\pi(\mfs) \in\mcA$, maximizes 
 the throughput defined as 
\begin{equation}\label{throughput.expr}
\eta(\pi)= \underset{N \rightarrow \infty}{\lim}\frac{1}{N}\sum\limits_{n=1}^N\Ex\big[ r(\mfs[n],\pi(\mfs[n])),\big],
\end{equation}
where the expectation is taken with respect to the random states $\mfs[1],\ld, \mfs[N]$ of the \gls{harq} process, and 
$r(\mfs[n],\pi(\mfs[n]))$ is the random \emph{reward} obtained using the actions $\pi(\mfs[n])$ after transmission of the block $n$  defined in \eqref{exp.rew}.

\subsection{Throughput optimization}\label{Sec:Th.MDP}
\subsubsection{Discretization}\label{sec.discri}
The problem is now formulated as the \gls{mdp} and, in order to make it tractable numerically we make the state-space $\mcS$ discrete. The first dimension $\mcCr$ of $\mcS$ is discrete by definition, so we discretize $\mcI$ using $T_{\mcI}$ points; which means that we need to discretize the set $[0,R[$ using $T_{\mcI}-1$ points (we used uniform quantization) and assign one discretization point to $\mcI_{R_+}$. 

Similarly, the actions set $\mcA$ is discretized over $T_{\tr{p}}$, thus  the encoding parameters $p\in\mcA_{\tr{p}}$ are discretized over $T_{\tr{p}}-2$ points. 

\subsubsection{Policy optimization}

Our goal is to solve the following optimization problem 
\begin{align}
\eta^*&=\max_{\pi \in\Pi} \eta(\pi),
\end{align}
where $\Pi $ is the set of admissible policies $\pi:\mcS\mapsto\mcA$. 

This average reward-per-stage problem can be solved using the so-called Bellman's equations \cite[Ch.~7.4]{Bertsekas05_book},
\begin{align}\label{bellman.eq}
\eta^*+h_{\mfs}&=\max_{a \in \mcA(\mfs)}\left[ r(\mfs,\mfa) + \sum_{\mfs'\in\mcS} p_{\mfs,\mfs'}(\mfa) h_{\mfs'} \right],\quad \forall \mfs\neq\mfs_\tr{p},\\
h_{\mfs_\tr{p}}&=0
\end{align}
thanks to the following lemma.

\begin{lemma}\label{assum1}
There is at least one state, denoted by $\mfs_\tr{p}$ and $m>0$, such that, for all initial states and all policies $\pi$, the probability of being in state $\mfs_\tr{p}$ at least once within the first $m$ times, is non zero, \ie $\Pr\set{ \mfs[n]=\mfs_\tr{p} }>0$, where $n<m$.
\end{lemma}
\begin{proof}
We take the special state $\mfs_\tr{p}=(1,0,I_\tr{p},0)$, where $I_\tr{p}$ is the smallest value defined by the discretization. Since $\nack$ message occurs with non-zero probability, and the probability of having arbitrarily small \gls{snr} is not zero, the probability of visiting the special state $\mfs_\tr{p}$ is non-zero, too.
\end{proof}

Under Lemma~\ref{assum1}, we obtain the guarantee that the optimal throughput $\eta^*$ is independent of the initial state \cite[Prop.~7.4.1.b]{Bertsekas05_book} and to solve equation \eqref{bellman.eq} for all $\mfs$, we may use two-step policy iteration algorithm for the average reward problem \cite[Ch.~7.4]{Bertsekas05_book}. In the first step, given the policy $\pi$, we calculate the corresponding average and differential rewards,  $\eta$ and $h_{\mfs}$, respectively,  that is, we solve the following equation for each $\mfs\neq\mfs_\tr{p}$ 
 \begin{align}\label{evaluate.policy}
  \eta+h_{\mfs}&=r(\mfs,\pi(\mfs)) + \sum_{\mfs'\in\mcS}p_{\mfs,\mfs'}\big(\pi(\mfs')\big) h_{\mfs'},
 \end{align}
where $h_{\mfs_\tr{p}}=0$.
In the next step, we perform a policy improvement to update $\pi(\mfs)$ as follows
 \begin{equation}
 \label{find.policy}
\pi(\mfs)\leftarrow\argmax_{\mfa \in \mcA(\mfs)}
\big[ r(\mfs,\mfa) + \sum_{\mfs'\in\mcS} p_{\mfs,\mfs'}(\mfa) h_{\mfs'} \big].
\end{equation}
The steps \eqref{evaluate.policy} and \eqref{find.policy} are repeated till convergence, which is guaranteed to be attained in finite number of iterations  \cite[Prop.~7.4.2]{Bertsekas05_book}.

In the numerical examples, the algorithm converges with a relatively small number of iterations ( $<4$) when we choose as initial policy $\pi[\mfs]=(1P,-)$. 

\subsection{Numerical Results}\label{Sec:Num.Results}
In this section we compare the performance of the proposed \gls{irharq}, \ie \TS~and~\SC, to the conventional \gls{irharq} in terms of attainable throughput as well as the outage probability. 

The throughput of the conventional \gls{irharq} can be calculated using renewal-reward theorem \cite{Caire01,Gan10} or, using our \gls{mdp} formulation, by considering that only  the ``conventional'' \gls{irharq} actions for all the states $\mfs \in   \mcS^{\tnr{1P}}_{\nack}$ are taken. That is, adopting the following policy:
\begin{align}\label{}
\pi(\mfs)=(\tnr{1P},-).
\end{align}

\subsubsection{\TS ~vs~ \SC}  First we investigate the encoding mode to be used, \ie we want to find the benefit of deciding in favor of  the \TS~or~\SC. To this end we first run the \gls{mdp} optimization for a fixed value of $R=4$ and analyze the case when $\mcA=\mcA_\tnr{mod}^{\TS}\times \mcA_{\tr{p}}$, that is, the \gls{harq} controller is able to choose among the modes \tnr{1P}, \tnr{0P} or \TS. These results, denoted by \TS, are shown in \figref{Fig:TH.R=4}. We also show therein, under the legend \SC, the results of the \gls{mdp} optimization when $\mcA=\mcA_\tnr{mod}^{\SC}\times \mcA_{\tr{p}}$, \ie the \gls{harq} controller is to choose one of the modes \tnr{1P}, \tnr{0P} or \SC. The throughput of the conventional \gls{irharq}, denoted as 1P, as well as the ergodic capacity $\ov{C}$ are also shown for comparison.

Using the conventional \gls{irharq} and for the fixed $R$ (\ie that does not change with $\SNRav$), the benefit of increasing the number of allowed transmission $\kmax$ materializes only for low \gls{snr} and thus, for small throughput values. This is why \gls{harq} is sometimes considered valuable only for low \gls{snr} regime. As we will see in \figref{Fig:f1}, similar value of throughput may be obtained decreasing $R$ and yet keeping $\kmax$ small. Therefore, from the system-level perspective, the most valuable throughput  gains are those obtained close to the nominal transmission rate $R$, where we see that the multi-packet \gls{harq} provides significant advantage over the conventional \gls{harq}. 

As a reference we consider the throughput $\eta=0.9R$  (shown by the horizontal dashed line in \figref{Fig:TH.R=4}). An important observation---already made in \cite{Larsson14}---is that the conventional \gls{harq} presents a large \gls{snr} gap to the ergodic capacity $\ov{C}$ when the transmission rate $R$ is fixed.  For instance, the gap of approximately $8.5$~dB can be observed in \figref{Fig:TH.R=4}. This puts in the perspective the often evoked property of \gls{irharq}, which says that the throughput of \gls{irharq} can attain the ergodic capacity
with infinite number of transmissions. While this is, indeed, true, this condition materializes for the throughput $\eta$ much smaller than $R$  or for a very small \gls{snr}. Alternatively, to reach the ergodic limit having the \gls{snr} fixed, a very large rate $R$  is needed. 

From that point of view, the proposed multi-packet transmission mode seems to reduce efficiently the gap. In fact, when $R=4$ bits/channel use and $\kmax=4$, the gap is reduced by more than $60$\%. A multi-packet HARQ with $\kmax=2$ can easily present a gain above $3$~dB compared to the conventional HARQ with larger $\kmax$. When $\kmax=3$, the gains are around $6$~dB; they increase negligible for $K=4$.

We also note that, even if the difference is relatively small (less than $1$~dB), the~\SC~always outperforms the~\TS. To obtain insight into the relevance of the encoding modes, considering $\mcA=\set{\tnr{1P}, \tnr{0P}, \TS, \SC }$ and $\kmax=2$ and defining the following probabilities conditioned on the retransmission being needed (\ie on $\mfs \in \mcS_{\ack,\nack}$)
\begin{align}
\label{P.1P.2}
P_\tnr{1P,$2$}&\triangleq\Pr \set{ \mfm[n]= \tnr{1P} | \mfs \in \mcS_{\ack,\nack}},\\
\label{P.SC}
P_\tnr{\SC}&\triangleq\Pr \set{ \mfm[n] = \SC | \mfs \in \mcS_{\ack,\nack}},\\
\label{P.TS}
P_\tnr{\TS}&\triangleq\Pr \set{\mfm[n] = \TS | \mfs \in \mcS_{\ack,\nack}},\\
\label{P.Drop}
P_\tnr{Drop}&\triangleq\Pr \set{ \mfm[n] = \tnr{0P}  |\mfs \in \mcS_{\ack,\nack}},\\
\label{P.ack}
P_\tnr{\ack,1}&\triangleq\Pr \set{ \mfs \in \mcS_{\ack,\ack}},
\end{align}
where
$P_\tnr{1P,$2$}$ is the probability of choosing one-packet retransmission, 
$P_\tnr{\SC}$ is the probability of choosing \gls{sc} encoding, 
$P_\tnr{\TS}$ is the probability of choosing \gls{ts} encoding, 
$P_\tnr{Drop}$ is the probability of ``dropping" the packet without retransmission  and 
$P_\tnr{\ack,1}$ is the probability that a packet is decoded after the first transmission. 

\begin{figure}[t]
\psfrag{SNR}[ct][ct][\siz]{$\SNRav$~[dB]}
\psfrag{throughput}[ct][ct][\siz]{$\eta$}
\psfrag{TS}[cl][cl][\siz]{\TS}
\psfrag{SC}[cl][cl][\siz]{\SC}
\psfrag{HA}[cl][cl][\siz]{\tnr{1P}}
\psfrag{ERG}[cl][cl][\siz]{$\ov{C}$}
\psfrag{K=2}[cl][cl][\siz]{$K=2$ }
\psfrag{K=3}[cl][cl][\siz]{$K=3$ }
\psfrag{K=4}[cl][cl][\siz]{$K=4$ }
\psfrag{n=0.9R=0.6}[cl][cl][\siz]{$\eta=0.9R=3.6$ }

\begin{center}
\scalebox{\sizf}{\includegraphics[width=\linewidth]{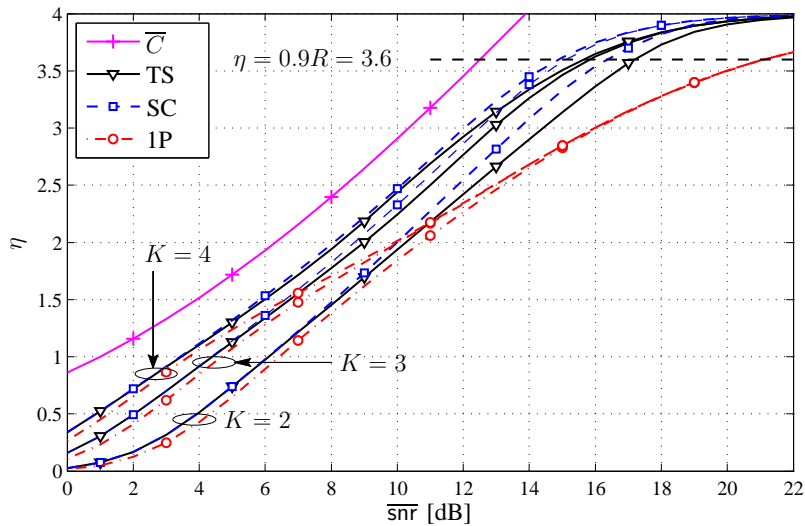}}
\caption{Throughput of the proposed multi-packet \gls{harq} ~\TS~and~ \SC~  with $T_{\tr{p}}=32$ are compared to the conventional \gls{harq} \tnr{1P} and the ergodic capacity $\ov{C}$ when $R=4$ bits/channel use and $\kmax=2, 3~\text{and}~4$.}
\label{Fig:TH.R=4}
\end{center}
\end{figure}

\begin{figure}[t]
\psfrag{XX}[l][l][\siz]{$P_\tnr{Drop}$}
\psfrag{SNR}[ct][ct][\siz]{$\SNRav$~[dB]}
\psfrag{RETR}[l][l][\siz]{$P_\tnr{1P,$2$}$}
\psfrag{SC}[l][l][\siz]{$P_\tnr{\SC}$}
\psfrag{NRS}[l][l][\siz]{$P_\tnr{\ack,1}$}
\psfrag{Probability}[l][l][\siz]{Probability}

\begin{center}
\scalebox{\sizf}{\includegraphics[width=\linewidth]{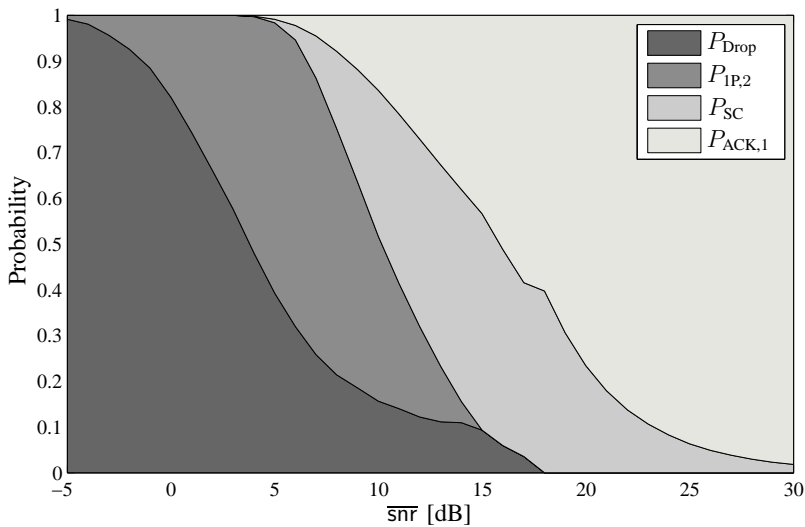}}
\caption{Probabilities defined in \eqref{P.1P.2}-\eqref{P.ack} with $R=4$ bits/channel use and $\kmax=2$.}
\label{Fig:actions.proba}
\end{center}
\end{figure}

\figref{Fig:actions.proba} shows the above-defined probability as a function of $\SNRav$ from which the following  observations can be made
\begin{itemize}
\item
Below a  threshold ($\SNRav=18$~dB), it is more profitable from the throughput point of view to drop the packet, which corresponds to the action $\mfa=(\tnr{0P} , -)$. 
\item 
The one-packet encoding dominates the multi-packet encoding for $\SNRav<10$~dB, which explains the throughput results are similar for the conventional and the proposed  multi-packet \gls{harq}. 
\item 
The multi-packet \gls{sc} transmission is likely to be used for $10~\tnr{dB}<\SNRav<25$~dB; this region of \gls{snr} corresponds also to the throughput of the multi-packet \gls{harq} (\SC) being significantly larger than the throughput  of the conventional \gls{harq} (\tnr{1P}). 
\item When \gls{sc} mode is available, the time-sharing mode is never used.
\item
Asymptoticaly, \ie increasing the \gls{snr}, we increase the probability of successful decoding in the first transmission, thus all \gls{harq} modes will offer a similar throughput for large \gls{snr}. 
\end{itemize}

\begin{figure}[b]
\psfrag{SNR}[ct][ct][\siz]{$\SNRav$~[dB]}
\psfrag{outage}[ct][ct][\siz]{$P_\tnr{out}$}
\psfrag{SC}[cl][cl][\siz]{\SC }
\psfrag{TS}[cl][cl][\siz]{\TS}
\psfrag{HAR}[cl][cl][\siz]{$\tnr{1P}$}
\psfrag{K=2}[cl][cl][\siz]{$K=2$ }
\psfrag{K=3}[cl][cl][\siz]{$K=3$ }
\psfrag{K=4}[cl][cl][\siz]{$K=4$ }
\begin{center}
\scalebox{\sizf}{\includegraphics[width=\linewidth]{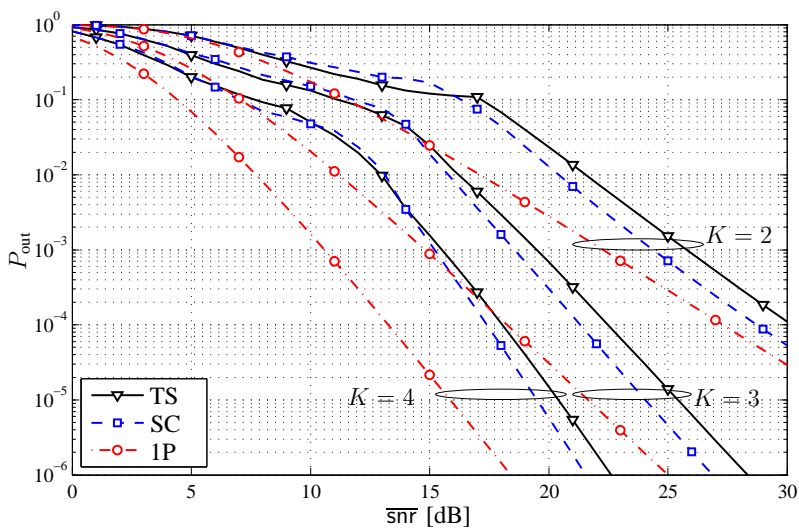}}
\caption{Outage corresponding to the one-packet and multiple-packet \gls{harq} when  $R=4$ bits/channel use and $T_{\tr{p}}=32$.}\label{Fig:Outage}
\end{center}
\end{figure}

The price to pay for the larger throughput is the increase in the outage as we illustrate in \figref{Fig:Outage}. While outage considerations were absent from our discussion, we note that it is also possible to design the policies which take into account the constraints on the outage, however, we leave this issue beyond the scope of our work.

In \figref{Fig:actions.par} we show an example of the optimal value of the parameter $p$ as a function of the \gls{ami} $I_{\hol,1}$, \ie when the optimal actions are $\mfa(\mfs )=(\TS,p)$ or $\mfa(\mfs )=(\SC,p)$ , and when $\bk^{\mfs}=(1,0)$. The intuition behind such results is clear: for larger $I_{\hol,1}$, \ie when the first transmission results are close to being decoded (this happens when $I_{\hol,1}>R=4$, the power (for \SC) or the time (for \TS) fractions attributed to the retransmission decrease. 
\begin{figure}[t]
\psfrag{SNR}[ct][ct][\siz]{$I_{{\hol},1}$}
\psfrag{policy}[ct][ct][\siz]{$p$}
\psfrag{SCX}[cl][cl][\siz]{\SC}
\psfrag{TSX}[cl][cl][\siz]{\TS}

\begin{center}
\scalebox{\sizf}{\includegraphics[width=\linewidth]{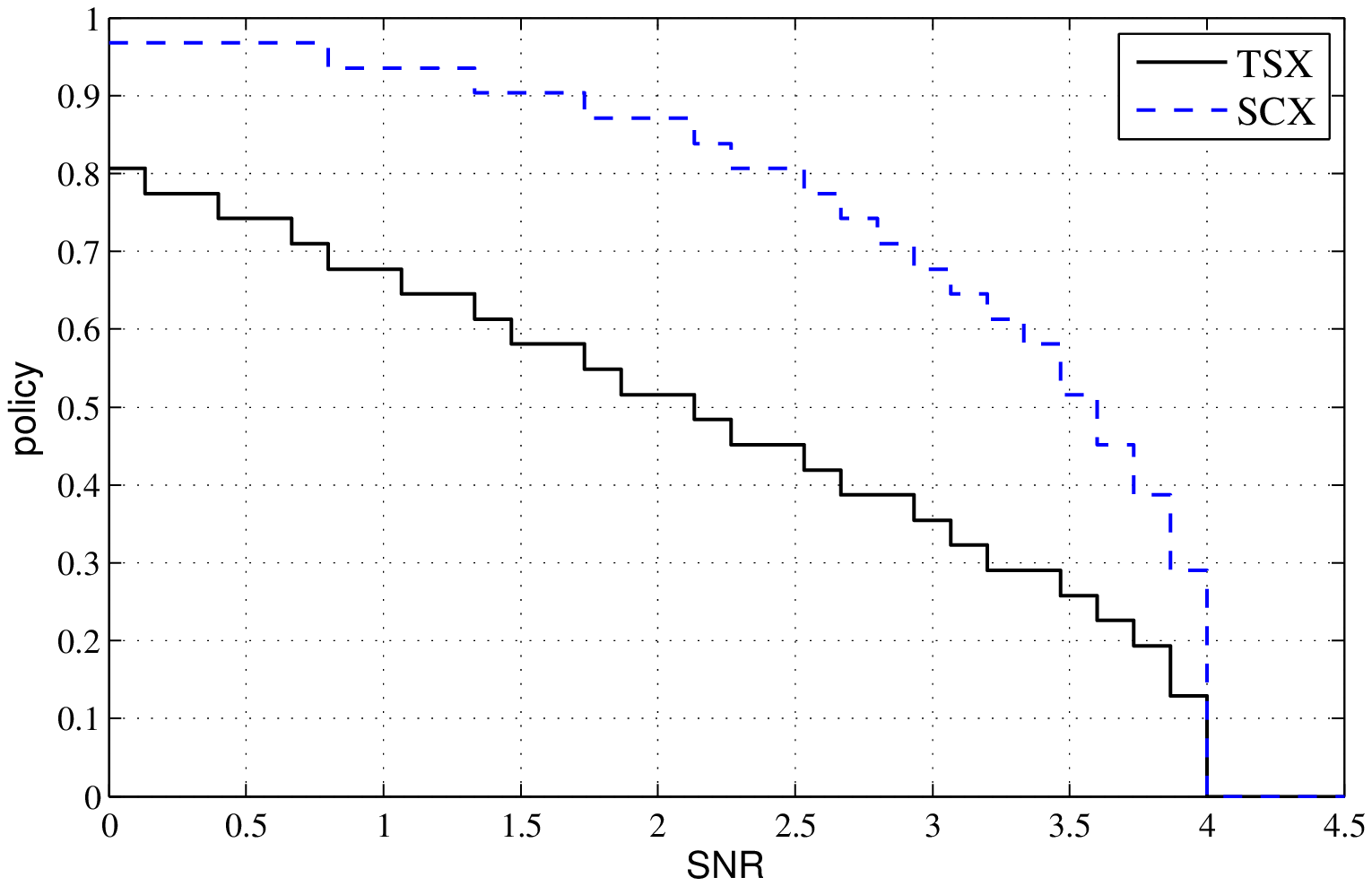}}
\caption{An example of the optimal parameter $p$ obtained for ~\TS~and~\SC~ when $R=4$ bits/channel use, $\SNRav=16$~dB, $T_{\tr{p}}=32$ and $\kmax=4$.}\label{Fig:actions.par}
\end{center}
\end{figure}

\subsubsection{Comparison for different $R$}
We show the results of the throughput for different values  $R$ in \figref{Fig:f1}; the results shown for $R=4$ are the same as those we already presented in \figref{Fig:TH.R=4}. As we can see, the gains of the~\SC~over the~\TS~are less pronounced for smaller values of $R$ and $\SNRav$. This is reminiscence of the similar behaviour of the conventional \gls{sc} broadcast transmission, where the gains with respect to the \gls{ts} appear also in high \gls{snr}. 

The main conclusion is that the multi-packet \gls{harq} provides an important increase of the throughput in the zone of interest (that is, for throughput values close to the nominal transmission rate $R$).

\begin{figure}[t]
\psfrag{XXXXXX1}[cl][cl][\siz]{\SC, $\kmax=2$ }
\psfrag{XXXXX2}[cl][cl][\siz]{\TS, $\kmax=2$}
\psfrag{XXXXX3}[cl][cl][\siz]{\tnr{1P}, $\kmax=2$}
\psfrag{XXXXX4}[cl][cl][\siz]{\tnr{1P}, $\kmax=4$}
\psfrag{SNR}[ct][ct][\siz]{$\SNRav$~[dB]}
\psfrag{T}[ct][ct][\siz]{$\eta$}
\psfrag{XX}[cl][cl][\siz]{$R=4$}
\psfrag{XX2}[cl][cl][\siz]{$R=2$}
\psfrag{XX3}[cl][cl][\siz]{$R=1$}
\psfrag{XX4}[cl][cl][\siz]{$R=0.5$}
\psfrag{ERG}[cl][cl][\siz]{$\ov{C}$}
\begin{center}
\scalebox{\sizf}{\includegraphics[width=\linewidth]{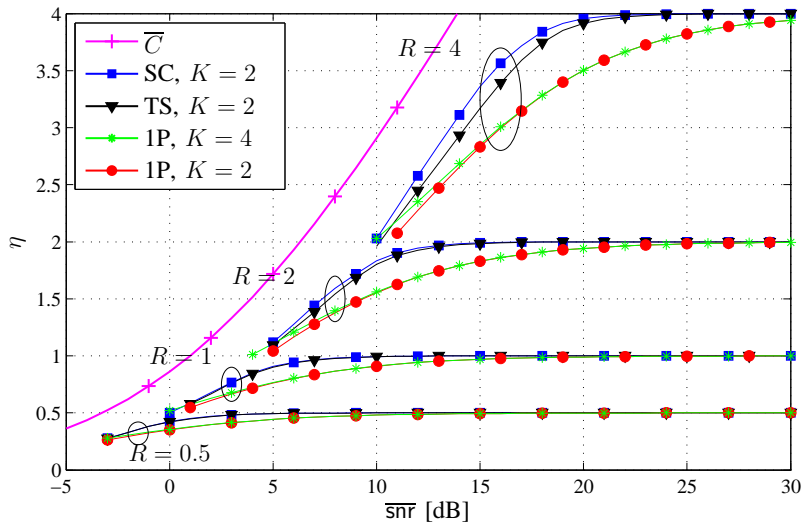}}
\caption{Throughput obtained for various values of $R$ bits/channel use, where for clarity, the results are shown only for  $\eta>R/2$.}\label{Fig:f1}
\end{center}
\end{figure}

\subsubsection{A Note on Discretization}

To provide an insight into the discretization effects, we show \figref{Fig:dis}. We emphasize here that we took  a sufficiently large $T_{\mcI}$ to accurately calculate the throughput ($T_{\mcI}=32$ in the numerical examples). Thus, discretization effects are almost entirely captured by $T_{\tr{p}}$. We note that \TS~and~\SC~ present notable gain compared to \tnr{1P} with only $T_{\tr{p}}=4$. The results do not change significantly for $T_{\tr{p}}\geq 16$. For a given $T_\tr{p}$, performance may be improved if we considered non uniform quantization, specially in the case of \SC; the issue of finding the optimal quantization is, however, out of scope of this work. 

\begin{figure}[t]
\psfrag{SNR}[ct][ct][\siz]{$\SNRav$~[dB]}
\psfrag{throughput}[ct][ct][\siz]{$\eta$}
\psfrag{TS,dis=32,K=4XXXXX}[cl][cl][\siz]{\TS, $T_{\tr{p}}=32, \kmax=4$ }
\psfrag{TS,dis=8,K=4}[cl][cl][\siz]{\TS, $T_{\tr{p}}=8, \kmax=4$ }
\psfrag{TS,dis=4,K=4}[cl][cl][\siz]{\TS,  $T_{\tr{p}}=4, \kmax=4$ }
\psfrag{HARQ,K=4}[cl][cl][\siz]{\tnr{1P}, $\kmax=4$ }
\psfrag{HARQ,K=8}[cl][cl][\siz]{\tnr{1P}, $\kmax=8$ }
\psfrag{SC,dis=32,K=4XXXXX}[cl][cl][\siz]{\SC, $T_{\tr{p}}=32, \kmax=4$ }
\psfrag{SC,dis=16,K=4}[cl][cl][\siz]{\SC, $T_{\tr{p}}=16, \kmax=4$ }
\psfrag{SC,dis=4,K=4}[cl][cl][\siz]{\SC, $T_{\tr{p}}=4, \kmax=4$ }
\psfrag{HARQ,K=4}[cl][cl][\siz]{\tnr{1P}, $\kmax=4$ }
\psfrag{HARQ,K=8}[cl][cl][\siz]{\tnr{1P}, $\kmax=8$ }
\psfrag{Erg}[cl][cl][\siz]{$\overline{\eta}$}

\begin{center}
\scalebox{\sizf}{\includegraphics[width=1\linewidth]{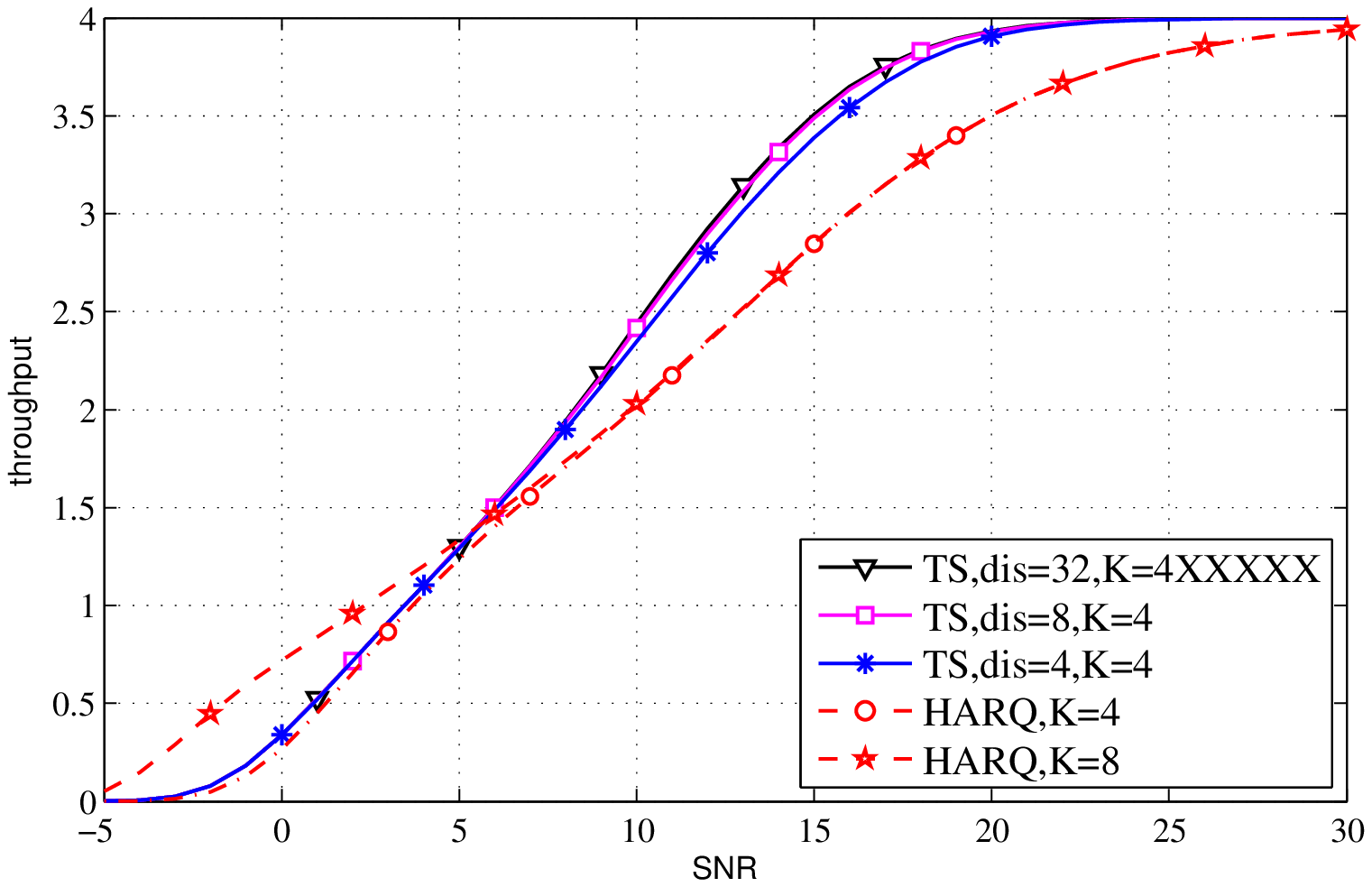}}
\\
(a)
\\
\scalebox{\sizf}{\includegraphics[width=1\linewidth]{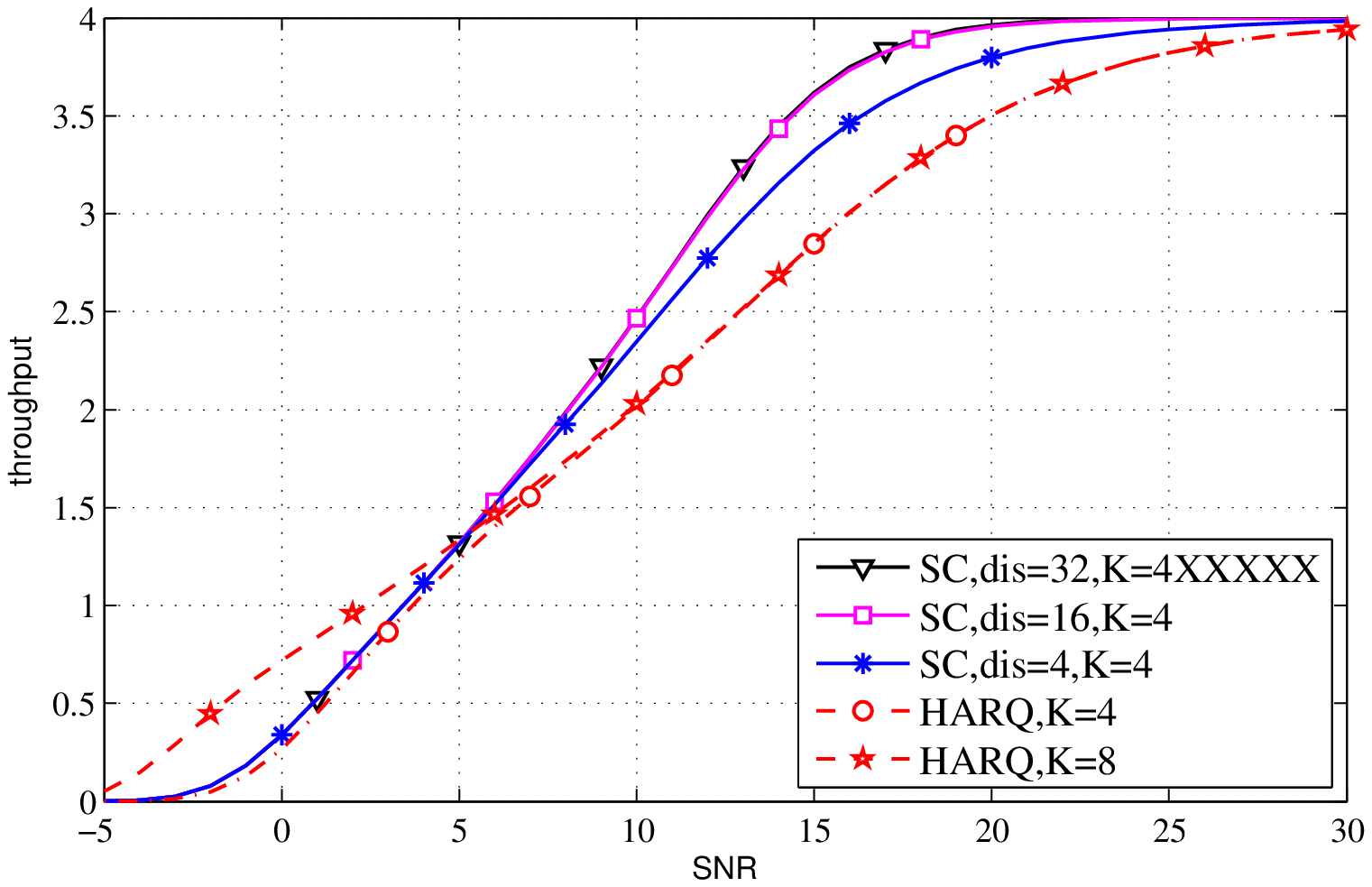}}
\\
(b)
\caption{Throughput of the proposed multi-packet in the case of (a) \gls{harq} with~\TS~ and (b) \gls{harq} with~\SC. Different values of $T_{\tr{p}}$ are considered and $R=4$ bits/channel use. The throughput of the conventional \gls{harq} \tnr{1P} is shown for comparison.} \label{Fig:dis}
\end{center}
\end{figure}
\section{One-bit feedback scenario} \label{Sec:One_bit}

In this section, we consider the scenario where only the conventional 1-bit feedback \ack/\nack~is available at the transmitter. In this case, the state space $\mcS=\mcCr \times \mcI^{2}$ is partially observable:  the random variable $\mcI^{2}$ is non-observable while $\mcCr$ is fully observable from transmitter by using only the received \ack/\nack~bits. This situation is known in the literature as \gls{psimdp} \cite{Arapostathis93}. 

In the  \gls{psimdp} context, the \gls{harq} controller decides to take  actions $\mfa[n]$ on the basis of \textit{observable history} defined as:
\begin{equation}\label{hist}
\mcO[n]=(\mfs[0],\mfa[0], \mcM[1], \bk[1] ,\mfa[1],\cdots, \mfa[n-1], \mcM[n], \bk[n])
\end{equation}
where the initial state $\mfs[0]$ and the corresponding action $\mfa[0]$ are assumed to be known. At the transmitter, the $\kk_{\ell}$ are updated using the \textit{observable} feedback messages $\mcM[n]$; the past actions $\mfa$ are also perfectly known to the transmitter. In general, not the entire history is useful  to the controller but only the parts related to the packets under transmission at time $n$, \ie \gls{hol} and  \gls{hol}-next packet. 

The standard procedure to solve \gls{psimdp}  problem consists in finding the equivalent perfect state \gls{mdp} problem $(\mcZ, \mcA, \mcW, Q', r')$ \cite{Bertsekas05_book}. The state space $\mcZ$ is $\mcCr \times \mcP(\mcI^{2})$, with $ \mcP(\mcI^{2})$, called \textit{the belief states}, is the space of all probability measures on $\mcI^{2}$. The definition of the action space $\mcA$ and the disturbance space $\mcW$ are the same as in the sec.~\ref{sec_multi.bit}. 

Let $\mfz\in \mcZ$ be defined as the couple $(\kk,b)$ where $b(x)=\pdf_{\Iv}(x/\mcO[n])$ is the  a posteriori distribution of $\Iv=(I_{{\hol},\kk_{\hol}},I_{{\hol+1},\kk_{\hol+1}})$ given an observable history $\mcO$. The statistical evolution of the system is captured by the transition law $Q' : \mcZ \times \mcA \mapsto \Pi(\mcZ)$ where $\Pi(\mcZ)$ presents the set of discrete probability distribution over  $\mcZ$. The expected reward $r'(\mfz,\mfa)$ for taking the action $\mfa$ in the state $\mfz$ is then given by:
\begin{equation}
r'(\mfz,\mfa)=\int_{ \mcI^{2}} \hat{r}(\mfz,\mfa,x)  b(x) \tr{d}x.
\end{equation}

The main challenge in solving the \gls{psimdp} problem is the characterization of the space of the belief states $\mcP(\mcI^{2})$, \ie the space of functions $b(x)$.

\subsection{Case  of maximum two allowed transmission}
When $\kmax=2$, we only need to track the value $I_{\hol,1}$ or $I_{\hol+1,1}$ when a \nack~message is received, \ie the \gls{ami} after the first transmission of the \gls{hol}, or \gls{hol}-next packet. Thus, the belief states $b(x)$ need to be defined over only one dimension and will be parametrized by the set of possible actions. In other words,  we define the state space using the actions  $\mfa[n-1]$ and the feedback message $\mcM[n]$. The closed form expression are given in \eqref{b.state.k.2} where $\mathbb{I}(x)=1$ if $x$  is true, and $0$ otherwise, $\cdf_{\SNR}(x)=1-\exp\big(- {x}/{\SNRav}\big)$ is the of $\SNR$' \gls{cdf} and $\Pr\{\nack\}$ is the probability of not decoding the superposed \gls{hol}-next packet given that the  \gls{hol} packet was also not decoded:
\begin{equation}
\Pr\{\nack\}=
\begin{cases} 
\cdf_{\SNR}(\frac{2^R-1}{1-p\cd2^R}), &~~\text{if}~~p< 2^{-R}\\
1,&~~\text{otherwise}
\end{cases}.
\end{equation}
\begin{figure*}

\begin{spacing}{1}
\begin{equation} \label{b.state.k.2}
\small{
b(x)=
\begin{cases} 
\displaystyle{ \frac{\log({2})\cd2^{\frac{x}{1-p}}\cd \pdf_{\SNR}\left(2^{\frac{x}{1-p}}-1\right)}{(1-p)\cd\cdf_{\SNR}\left(2^{\frac{R}{1-p}}-1\right)}} \cd \mathbb{I}(x\leq R)&~ ~\text{if}~~\mcM[n]\in \set{(\ack,\nack),(\nack,\nack)},\\
&~~~\bk[n]=(2,1), \mfm[n]=\TS, \mfm[n+1]\in \set{\TS, \tnr{1P}}\\
\\
\displaystyle{ \frac{\log({2})\cd(1-p)\cd 2^{x}}{(1-p\cd2^x)^2\cd \Pr\{\nack\}}}\cd\pdf_{\SNR}\left(\frac{2^x-1}{1-p\cd 2^x}\right) \cd \mathbb{I}(p< 2^{-x})  &~~\text{if}~~\mcM[n]=(\nack,\nack), \bk[n]=(2,1),\\
 &~~~~~ ~\mfm[n]=\SC, \mfm[n+1]\in\set{\SC, \tnr{1P}}\\
\\
\displaystyle{ \frac{\log({2})\cd2^{x}}{(1-p)\cd \cdf_{\SNR}\left(\frac{2^{R}-1}{1-p}\right)}}\cd\pdf_{\SNR}\left(\frac{2^x-1}{1-p}\right) \cd \mathbb{I}(x\leq R) &~~\text{if}~~\mcM[n]=(\ack,\nack), \bk[n]=(2,1),\\
 &~~~~~ ~\mfm[n]=\SC, \mfm[n+1]\in\set{\SC, \tnr{1P}}\\
\\
\displaystyle{ \frac{ \log(2)\cd 2^{x}\cd \pdf_{\SNR}\left(2^{x}-1\right)}{\cdf_{\SNR}\left(2^{R}-1\right)}} \cd \mathbb{I}(x\leq R)&~~\text{if}~~\mcM[n]=(\nack,-), \bk[n]=(1,0)\\
\end{cases}}
\end{equation}

\end{spacing}
\hrulefill
\end{figure*}

\subsection{General case}

When $\kmax>2$, the belief states $b(x)$ are defined over two dimensions and cannot be derived in closed form. 

When the belief states are defined over one dimension, an approximation method was proposed in \cite{these_tajan}, which projects the beliefs on the parametrized set of functions. For example, Beta$(\theta_{1},\theta_{2})$ function was used in \cite{these_tajan} to parametrize the \gls{ami} $I_{\hol,\kk_\hol}$ (\cite{these_tajan} considered one-packet \gls{harq}). In our case, the observations $I_{\hol,\kk_{\hol}}$ and $I_{\hol+1,\kk_{\hol+1}}$ are dependant and we would need a projection of the joint \gls{pdf} of these two variables on the space of two-dimensional functions. This approach is thus tedious and to overcome this difficulty, we assume that if a retransmission is needed, the controller will always adopt the unique-action policy:
\begin{align}\label{alloc_app}
\pi(\mfs)=\hat{\mfa}.
\end{align}
Formally, the objective is to solve:
\begin{align}
\tilde{\eta}&=\max_{\hat{\mfa} \in\mcA} \eta(\pi),
\end{align}
and an exhaustive research over the one-dimensional space of allowed actions $\mcA$ is sufficient to determine the suboptimal action.  

\subsection{Numerical Results}
\begin{figure}[t]
\psfrag{SNR}[ct][ct][\siz]{$\SNRav$~[dB]}
\psfrag{throughput}[ct][ct][\siz]{$\eta$}
\psfrag{TSA}[cl][cl][\siz]{$\tilde{\eta}$ }
\psfrag{TSD}[cl][cl][\siz]{\TS}
\psfrag{TS1}[cl][cl][\siz]{$\hat{\eta}$}
\psfrag{HAR}[cl][cl][\siz]{${\tnr{1P}}$}
\psfrag{K=2}[cl][cl][\siz]{$K=2$ }
\psfrag{K=3}[cl][cl][\siz]{$K=3$ }
\psfrag{K=4}[cl][cl][\siz]{$K=4$ }
\psfrag{SCA}[cl][cl][\siz]{$\tilde{\eta}$ }
\psfrag{SCS}[cl][cl][\siz]{\SC}
\psfrag{SC1}[cl][cl][\siz]{$\hat{\eta}$}
\psfrag{HAR}[cl][cl][\siz]{${\tnr{1P}}$}

\begin{center}
\scalebox{\sizf}{\includegraphics[width=1\linewidth]{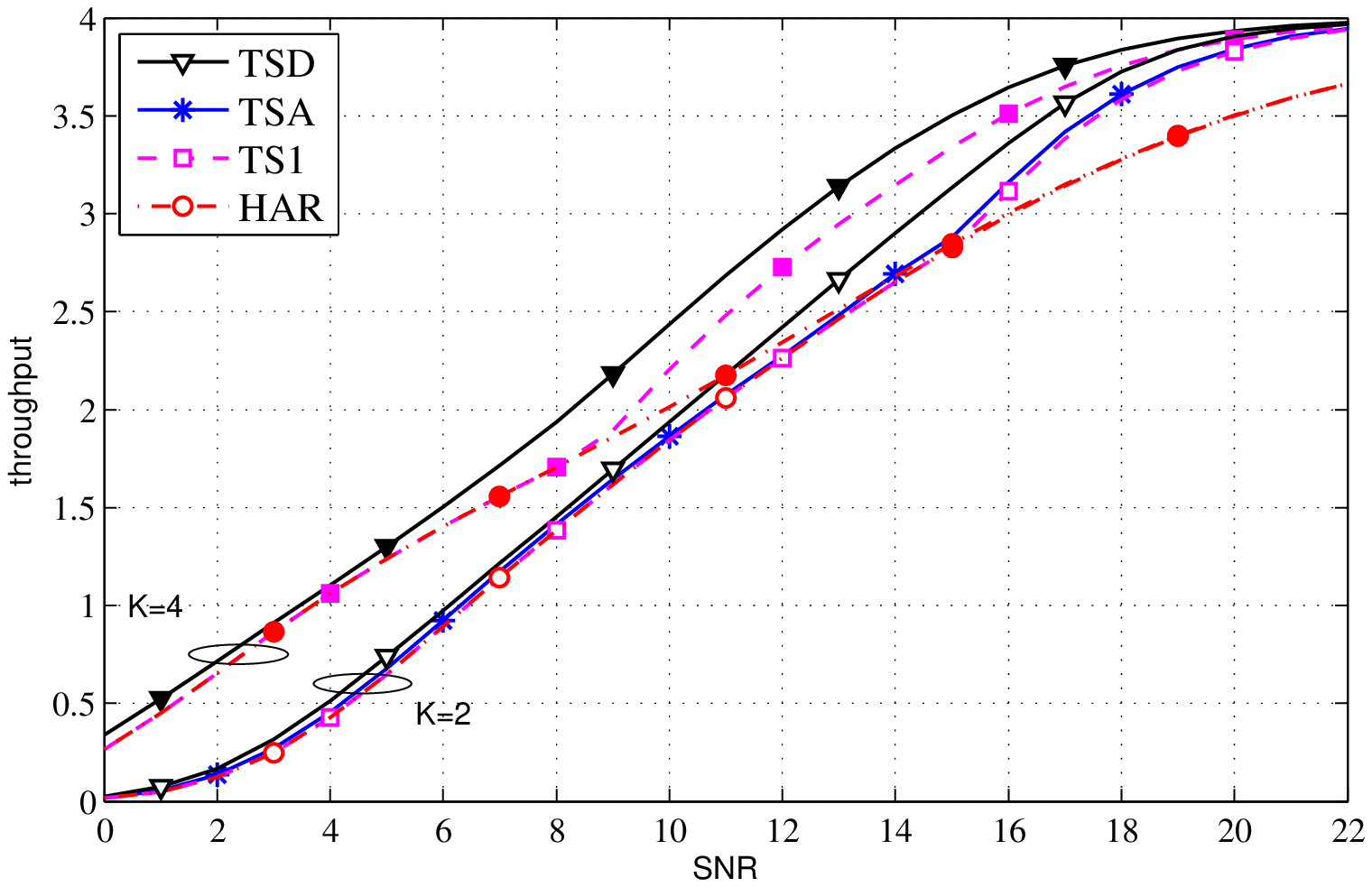}}
\\
(a)
\\
\scalebox{\sizf}{\includegraphics[width=1\linewidth]{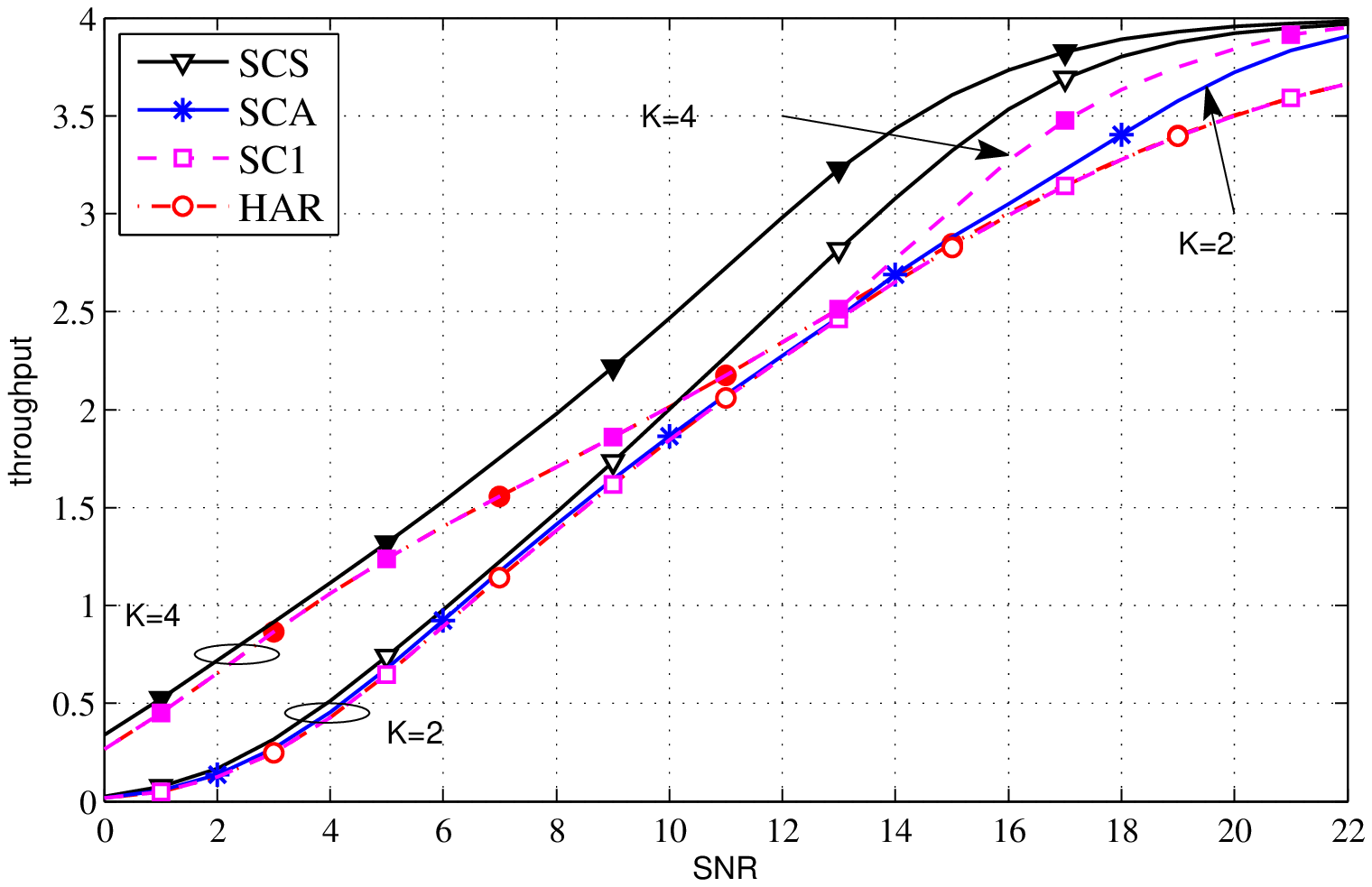}}
\\
(b)
\caption{Throughput of the multi-packet \gls{harq} in the case of (a) $\mcA=\mcA_\tnr{mod}^{\TS}\times \mcA_{\tr{p}}$ and (b) $\mcA=\mcA_\tnr{mod}^{\SC}\times \mcA_{\tr{p}}$; the multi-bit feedback (TS), one-bit feedback ($\tilde{\eta}$), unique-action one-bit feedback ($\hat{\eta}$) and conventional one-packet \gls{harq} (1P). $T_{\tr{p}}=16$ and $R=4$ bits/channel use.} \label{Fig:all}
\end{center}
\end{figure}

\figref{Fig:all} compare the performance of the proposed multi-packet \gls{harq} protocols. When the whole history $\mcO$ is used, the results are denoted by $\tilde{\eta}$ (for $K=2$), while $\hat{\eta}$ presents the simplified policy in \eqref{alloc_app}.

We observe that
\begin{itemize}
 \item For~\TS, the simplified unique-action protocol defined in \eqref{alloc_app} yields practically the same throughput as the one based on the complete parametrization of the state-space (for $K=2$). We thus conjecture that the same results will be obtained for $K>2$.
\item As expected, using one-bit feedback introduces the penalty with respect to multi-bit feedback but still, using \gls{ts}, the gains over the conventional \gls{harq} are notable, varying from $2.5$~dB (for $K=2$) to $4$~dB (for $K=4$).
\item For~\SC, the simplified unique-action protocol yields the same results as the one-packet transmission for $\kmax=2$; this is not entirely surprising as we already observed that the \gls{sc} is very sensitive to the discretization of the parameter space $A_\tr{p}$ in \figref{Fig:dis} (b). For $K=4$, we obtain an appreciable gain of $3$~dB.
\item When the complete parametrization of the state-space is used (for $K=2$), the gains of ~\SC~over the conventional \gls{harq} are around $1.5$ dB.
\end{itemize}

\section{Conclusions}\label{Sec:Conclusions}
In this work we proposed and analyzed the so-called multi-packet \gls{harq}, where various packets are simultaneously transmitted within the same block. We consider adjusting the joint encoding modes depending on the state of the receiver in the past transmissions and, formulating the problem as partial or full state information Markov decision process, we optimize the encoding modes and parameters of the \gls{harq}. Our results indicate that the joint encoding yields gains of various dB over the conventional \gls{harq} even in the simple case of one-bit feedback and \gls{harq} truncated to $\kmax=2$ transmission. These gains can be increased by $1-2$~dB by adding a few ($3-4$) additional feedback bits.

We also observed that the gains of the \gls{sc} with respect to the \gls{ts} are relatively small in the case of full observable state; one-bit feedback, however, it removes the advantage of the \gls{sc}. Thus, in the point-to-point \gls{harq}, the \gls{ts}, being simpler to implement, should be preferred over \gls{sc}. 

\section*{Appendix A}
\label{Sec:AppendixA}

We aim at determining the expression of $p_{\mfs,\mfs'}$ when $\kmax=2$. In this case, the set of possible values of $\bk=(\kk_{\hol}, \kk_{\hol+1})$ is $\mcCr=\set{(1,0),(2,1),(2,0)}$. One possible illustration of the state space $\mcS$ is given in \figref{Fig:markov-chain}. For each possible value of $\bk$, the vertical line presents the possible values of $I_{{\hol},\kk_{\hol}}$ while the horizontal line presents $I_{{\hol+1},\kk_{\hol+1}}$. When only the \gls{hol} packet is transmitted, \ie $I_{{\hol+1},\kk_{\hol+1}}=0$, the  horizontal line is irrelevant. For convenience we refer to the discretization interval corresponding to a state $\mfs$ as $\mcI^\mfs$.

The state $s$ in the figure corresponds to parameter $\bk^{\mfs}=(1,0)$ and $\bI^{\mfs}=(I_{\hol,1}^{\mfs},0)$. The transition from the state $\mfs$ at time $n$ to the state $\mfs'$ at time $n+1$ depends on the action $\mfa[n]=(\mfm[n],p[n])$ and the \gls{snr}. For example three situations can occur regarding the value of $\mfm[n]$. Namely: 
\begin{itemize}
\item If $\mfm[n]=\tnr{1P}$ the system moves to $\bk^{\mfs'}=(2,0)$. This transition, represented by a solid arrow, has the probability $p_{\mfs,\mfs'}(\mfa)=\Pr \set{I^{\mfs}_{{\hol},1}+ C(\SNR) \in \mcI_{{\hol},2}^{\mfs'}}$.
\item If $\mfm[n]=\tnr{0P}$ the \gls{harq} controller increments the \gls{hol} according to \eqref{upd.rule.joint} and the system moves to $\bk^{\mfs'}=(1,0)$, which is presented by the doted arrow. In this case, $p_{\mfs,\mfs'}(\mfa)=\Pr \set{ C(\SNR) \in \mcI_{{\hol},1}^{\mfs'}}$.
\item If $\mfm[n]=\TS$ the multi packet scenario is adopted and the counters are updated to $\bk^{\mfs'}=(2,1)$. This situation is depicted by the dashed arrow and $p_{\mfs,\mfs'}(\mfa)=\Pr \set{I^{\mfs}_{{\hol},1}+p C(\SNR) \in \mcI^{\mfs'}_{{\hol},2}  ~\cap~  (1-p) C(\SNR) \in \mcI^{\mfs'}_{{\hol+1},1}}$.
\end{itemize}
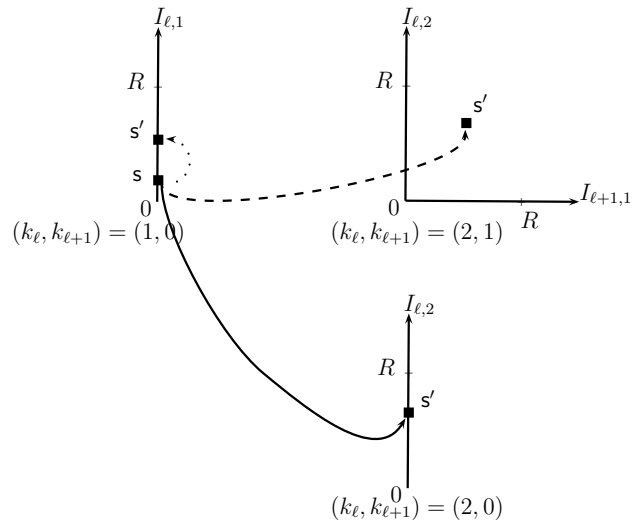
\begin{figure}[t]
\begin{center}
\scalebox{0.75} 
{
\begin{pspicture}(0,-4.8001366)(10.061289,4.8001366)
\psline[linewidth=0.04cm,arrowsize=0.05291667cm 2.0,arrowlength=1.4,arrowinset=0.4]{->}(6.2452345,1.0968555)(6.2652345,4.1968555)
\psline[linewidth=0.04cm,arrowsize=0.05291667cm 2.0,arrowlength=1.4,arrowinset=0.4]{->}(6.2452345,1.0968555)(9.325234,1.0768554)
\psdots[dotsize=0.12,dotstyle=+](6.2652345,3.1368554)
\psdots[dotsize=0.12,dotstyle=+](8.305234,1.0768554)
\usefont{T1}{ptm}{m}{n}
\rput(5.8637986,3.2418554){$R$}
\usefont{T1}{ptm}{m}{n}
\rput(8.443799,0.7618555){$R$}
\usefont{T1}{ptm}{m}{n}
\rput(6.4532814,4.32){$I_{{\hol},2}$}
\usefont{T1}{ptm}{m}{n}
\rput(9.7987795,1.2018554){$I_{{\hol+1},1}$}
\usefont{T1}{ptm}{m}{n}
\rput(6.0423145,0.98185545){$0$}
\psline[linewidth=0.04cm,arrowsize=0.05291667cm 2.0,arrowlength=1.4,arrowinset=0.4]{->}(1.8452344,1.0768554)(1.8652344,4.1768556)
\psdots[dotsize=0.12,dotstyle=+](1.8652344,3.1168554)
\usefont{T1}{ptm}{m}{n}
\rput(1.4637989,3.2218554){$R$}
\usefont{T1}{ptm}{m}{n}
\rput(2.0543358,4.32){$I_{{\hol},1}$}
\usefont{T1}{ptm}{m}{n}
\rput(1.6423144,0.9618555){$0$}
\psline[linewidth=0.04cm,arrowsize=0.05291667cm 2.0,arrowlength=1.4,arrowinset=0.4]{->}(6.2852345,-4.0031447)(6.3052344,-0.90314454)
\psdots[dotsize=0.12,dotstyle=+](6.3052344,-1.9631445)
\usefont{T1}{ptm}{m}{n}
\rput(5.903799,-1.8581445){$R$}
\usefont{T1}{ptm}{m}{n}
\rput(6.494336,-0.77){$I_{{\hol},2}$}
\usefont{T1}{ptm}{m}{n}
\rput(6.0823145,-4.1181445){$0$}
\usefont{T1}{ptm}{m}{n}
\rput(0.8,0.5218555){$(\kk_{\hol}, \kk_{\hol+1})=(1,0)$}
\usefont{T1}{ptm}{m}{n}
\rput(6.4504104,0.5218555){$(\kk_{\hol}, \kk_{\hol+1})=(2,1)$}
\usefont{T1}{ptm}{m}{n}
\rput(6.4504104,-4.35){$(\kk_{\hol}, \kk_{\hol+1})=(2,0)$}
\psdots[dotsize=0.2,dotstyle=square*](1.8652344,2.1768556)
\psdots[dotsize=0.2,dotstyle=square*](7.3252344,2.4768555)
\psdots[dotsize=0.2,dotstyle=square*](1.8652344,1.4568554)
\psdots[dotsize=0.2,dotstyle=square*](6.3052344,-2.6631446)
\psbezier[linewidth=0.04,arrowsize=0.05291667cm 2.0,arrowlength=1.4,arrowinset=0.4]{->}(1.9252343,1.3768555)(1.9252343,0.5768555)(2.9552295,-1.3251064)(3.7252343,-1.9631445)(4.4952393,-2.6011827)(5.734739,-3.623025)(6.2452345,-2.7631445)
\psbezier[linewidth=0.04,linestyle=dashed,dash=0.16cm 0.16cm,arrowsize=0.05291667cm 2.0,arrowlength=1.4,arrowinset=0.4]{->}(1.9252343,1.3968555)(1.9252343,0.59685546)(7.3052344,1.5568554)(7.3052344,2.3568554)
\rput{-89.88666}(0.22477505,3.8172667){\psarc[linewidth=0.04,linestyle=dotted,dotsep=0.16cm,arrowsize=0.05291667cm 2.0,arrowlength=1.4,arrowinset=0.4]{->}(2.0248003,1.7960232){0.39936668}{0.0}{183.60481}}
\usefont{T1}{ptm}{m}{n}
\rput(7.581504,2.8218555){$\mfs'$}
\usefont{T1}{ptm}{m}{n}
\rput(6.650205,-2.4181445){$\mfs'$}
\usefont{T1}{ptm}{m}{n}
\rput(1.4944726,2.3618555){$\mfs'$}
\usefont{T1}{ptm}{m}{n}
\rput(1.4944726,1.5418555){$\mfs$}
\end{pspicture} 
}
\caption{State space $\mcS$ representation when $\kmax=2$}\label{Fig:markov-chain}
\end{center}
\end{figure}




\balance

\bibliographystyle{IEEEtran}

\end{document}